\documentclass[11pt]{article}

\usepackage[letterpaper, margin=1in]{geometry}

\usepackage[parfill]{parskip}

\usepackage[utf8]{inputenc}

\usepackage{varioref}
\usepackage[linkcolor=black,colorlinks=true,citecolor=black,filecolor=black]{hyperref}
\usepackage{comment}

\usepackage{mathtools}
\usepackage{amsmath, amssymb, amsfonts, mathrsfs}
\usepackage[sc]{mathpazo} 
\usepackage{amsthm}
\usepackage{subcaption}
\usepackage[ruled, vlined]{algorithm2e}
\usepackage{aligned-overset}
\usepackage{cleveref}
\usepackage{authblk}
\usepackage{thmtools}
\usepackage{thm-restate}
\usepackage{nicefrac}
\usepackage[only,llbracket,rrbracket]{stmaryrd}
\usepackage{orcidlink}

\theoremstyle{plain}
\numberwithin{equation}{section}
\newtheorem{theorem}{Theorem}[section]

\newtheorem{corollary}[theorem]{Corollary}
\newtheorem{definition}[theorem]{Definition}
\newtheorem{lemma}[theorem]{Lemma}

\newtheorem{observation}[theorem]{Observation}

\crefname{corollary}{Corollary}{Corollaries}
\crefname{lemma}{Lemma}{Lemmata}

\newcommand{\N}{\mathbb{N}}
\newcommand{\Q}{\mathbb{Q}}
\newcommand{\R}{\mathbb{R}}

\newcommand{\bigO}{\mathcal{O}}
\DeclareMathOperator{\poly}{poly}

\let\emptyset\varnothing

\renewcommand{\epsilon}{\ensuremath\varepsilon}

\renewcommand{\phi}{\ensuremath{\varphi}}

\newcommand{\proj}{\text{proj}}

\DeclareMathOperator{\vol}{vol}

\newcommand{\nonexp}{\mathsf{NonExp}}
\newcommand{\contr}{\mathsf{Contr}}
\newcommand{\One}{\mathbf{1}}

\begin{document}
\author{Andrei Feodorov\orcidlink{0009-0006-3367-1371}}
\author{Sebastian Haslebacher\orcidlink{0000-0003-3988-3325}}
\affil{Department of Computer Science, ETH Zurich, Switzerland}
\title{Faster Approximate Fixed Points of $\ell_\infty$-Contractions}
\date{}
\maketitle

\begin{abstract}
    We present a new algorithm for finding an $\epsilon$-approximate fixed point of an $\ell_\infty$-contracting function $f : [0, 1]^d \rightarrow [0, 1]^d$. Our algorithm is based on the query-efficient algorithm by Chen, Li, and Yannakakis~(STOC 2024), but comes with an improved upper bound of $(\log \frac{1}{\epsilon})^{\bigO(d \log d)}$ on the overall runtime (while still being query-efficient). By combining this with a recent decomposition theorem for $\ell_\infty$-contracting functions, we then describe a second algorithm that finds an $\epsilon$-approximate fixed point in $(\log \frac{1}{\epsilon})^{\bigO(\sqrt{d} \log d)}$ queries and time. The key observation here is that decomposition theorems such as the one for $\ell_\infty$-contracting maps often allow a trade-off: If an algorithm's runtime is worse than its query complexity in terms of the dependency on the dimension $d$, then we can improve the runtime at the expense of weakening the query upper bound. By well-known reductions, our results imply a faster algorithm for $\epsilon$-approximately solving Shapley stochastic games.
\end{abstract}

\renewcommand{\abstractname}{Acknowledgements}
\begin{abstract}
    We would like to thank Bernd Gärtner, Patrick Schnider, and Simon Weber from ETH Zurich for valuable discussions and feedback. We would also like to thank Hanwen Zhang from the University of Copenhagen for pointing us to useful results on hyperplane arrangements and for suggesting an improvement in the proof of \Cref{lemma:balancing_pyramids}.
\end{abstract}

\newpage

\section{Introduction}
\label{sec:introduction}

Consider a function $f : [0, 1]^d \rightarrow [0, 1]^d$ defined on the unit cube in $\R^d$. We say that $f$ is $\lambda$-contracting in the $\ell_\infty$-metric if it satisfies $\|f(x) - f(y)\|_\infty \leq \lambda \|x - y\|_\infty$ for all $x, y \in [0, 1]^d$ and some fixed $\lambda < 1$, where $\|x\|_\infty := \max_{i \in [d]} |x_i|$. Given query access to $f$, we consider the problem of algorithmically finding an $\epsilon$-approximate fixed point of $f$, i.e.\ a point $x$ satisfying $\|x - f(x)\|_\infty \leq \epsilon$ for some $\epsilon > 0$ that is considered part of the input. Note that this problem is well-defined by Banach's fixed point theorem~\cite{banach1922operations}, which tells us that $f$ must have a unique fixed point $x^\star = f(x^\star)$. 

Banach's proof~\cite{banach1922operations} implies the following iterative algorithm for this $\ell_\infty$-contraction problem: Start at some point $x \in [0, 1]^d$ and iteratively update $x \gets f(x)$ until the condition $\|x - f(x)\|_\infty \leq \epsilon$ is satisfied. A simple argument involving the contraction property shows that this algorithm will terminate after at most $\bigO \left( \log(\frac{1}{\epsilon}) / \log (\frac{1}{\lambda}) \right)$ iterations. 

Unfortunately, there are applications where $\lambda$ can be very close to $1$, in which case Banach's iterative algorithm is very slow. For example, both Condon's simple stochastic games~\cite{condonComplexityStochasticGames1992} and Shapley's stochastic games~\cite{shapley1953stochastic} can be reduced to the problem of finding an $\epsilon$-approximate fixed point of a $\lambda$-contracting map with $\frac{1}{1 - \lambda}$ exponential in their respective input size~(see e.g.\@~\cite{etessamiTarskiTheoremSupermodular2020}). Consequently, using Banach's iterative algorithm for these games yields worst-case guarantees that are exponential in their input size.

Applications such as the ones mentioned above motivate the search of alternative algorithms with a runtime independent of $\lambda$. For example, Shellman and Sikorski~\cite{shellmanTwoDimensionalBisectionEnvelope2002} gave an algorithm that runs in time $\bigO(\log \frac{1}{\epsilon})$ in the two-dimensional case $d = 2$. Indeed, their runtime is independent of $\lambda$ and in fact, their algorithm works even if $\lambda = 1$ (note that in this case, $f$ is still guaranteed to have a fixed point by Brouwer's fixed point theorem~\cite{brouwerUeberAbbildungMannigfaltigkeiten1911}, but this fixed point is not necessarily unique anymore). They later generalized this, obtaining an algorithm for $d$-dimensional instances that runs in $\bigO(\log^d (\frac{1}{\epsilon}))$ time~\cite{shellmanRecursiveAlgorithmInfinitynorm2003}. Very recently, this was improved by Chen, Li, and Yannakakis~\cite{chenQuadraticSpeedupComputing2026} to $\bigO(\log^{\lceil d / 2 \rceil} (\frac{1}{\epsilon}))$, which is the current state-of-the-art in terms of time complexity for the $\ell_\infty$-contraction problem when $\lambda$ is very close or equal to $1$. If we additionally assume $f$ to be monotone (which is also motivated by the above mentioned applications), a slightly better bound is known: Indeed, Batziou, Fearnley, Gordon, Mehta, and Savani~\cite{batziouMonotoneContractions2025} showed that an algorithm running in time $\bigO((c \log \frac{1}{\epsilon})^{\lceil d / 3 \rceil})$ for some constant $c$ is possible in this case.

While all these algorithms got rid of the dependency on $\lambda$, their runtime now unfortunately depends exponentially on the dimension $d$. As a consequence, applying them to Condon's or Shapley's stochastic games still only yields exponential-time algorithms in the worst case. However, a recent breakthrough by Chen, Li, and Yannakakis~\cite{chenComputingFixedPoint2025} suggests that much better algorithms might be possible: Indeed, they gave a (very slow) algorithm that finds an $\epsilon$-approximate fixed point of $f$ while only querying $f$ at most $\bigO(d \log \frac{1}{\epsilon})$ times. This is particularly interesting because in all of the previous algorithms, the runtime is essentially dominated by the number of queries made to $f$ (which is exponential in $d$). The result in~\cite{chenComputingFixedPoint2025} shows that the query complexity of the $\ell_\infty$-contraction problem is small, and thus not a barrier towards potentially having algorithms that run e.g.\ in time $\poly(d, \log \frac{1}{\epsilon})$. 

Of course, the drawback of the algorithm in~\cite{chenComputingFixedPoint2025} is its runtime, which is much worse than the runtime of the previous algorithms: Indeed, while the algorithm makes only few queries to $f$, each query point in their algorithm is determined by a brute-force search running in time $\Omega( (\frac{1}{\epsilon})^d)$. Our main technical contribution is to address this drawback by replacing the brute-force search with a new procedure, yielding the following theorem.

\begin{restatable}{theorem}{thmmainresultone}
\label{thm:main_result_1}
    Given $\epsilon > 0$ and query-access to a $\lambda$-contracting (in the $\ell_\infty$-metric) function $f : [0, 1]^d \rightarrow [0, 1]^d$, there is an algorithm that finds an $\epsilon$-approximate fixed point $x \in [0, 1]^d$ (i.e.\ with $\|x - f(x)\|_\infty \leq \epsilon$) in $\bigO(d^2 \log \frac{1}{\epsilon})$ queries and $(\log \frac{1}{\epsilon})^{\bigO(d \log d)}$ time. 
\end{restatable}

Observe that neither the query nor the runtime bounds in \Cref{thm:main_result_1} beat the respective state-of-the-art bounds: Indeed, most of the existing algorithms including the $\bigO(\log^{\lceil d / 2 \rceil} \frac{1}{\epsilon})$-time algorithm in~\cite{chenQuadraticSpeedupComputing2026} are faster than ours, and the algorithm in~\cite{chenComputingFixedPoint2025} uses fewer queries. However, our algorithm crucially achieves comparable performance in both aspects simultaneously: It is still query-efficient, and has a runtime that is not too far off the state-of-the-art. 

It turns out that this improvement has important consequences. Indeed, one of our key observations is that \Cref{thm:main_result_1} is sufficient to allow for an effective combination with a recent decomposition theorem of Chen, Li, and Yannakakis~\cite{chenQuadraticSpeedupComputing2026}. On a high level, their decomposition theorem allows us to trade off our algorithm's query complexity for its time complexity in terms of the dependency on the dimension $d$. Since the only bad part in the runtime of \Cref{thm:main_result_1} is its dependency on $d$, we can use this to prove the following theorem.

\begin{restatable}{theorem}{thmmainresulttwo}
\label{thm:main_result_2}
    Given $\epsilon > 0$ and query-access to a $\lambda$-contracting function $f : [0, 1]^d \rightarrow [0, 1]^d$, there is an algorithm that finds a point $x \in [0, 1]^d$ with $\|x - f(x)\|_\infty \leq \epsilon$ in $(\log \frac{1}{\epsilon})^{\bigO(\sqrt{d} \log d)}$ time and queries. 
\end{restatable}

Note that our algorithm in \Cref{thm:main_result_2} is not query-efficient anymore.

\paragraph*{Implications for Condon's and Shapley's Stochastic Games}

Condon's simple stochastic games are played on a graph with $n$ vertices and can be solved in expected subexponential time $2^{\bigO(\sqrt{n} \log n)} \poly(|G|)$~\cite{ludwigSubexponentialRandomizedAlgorithm1995, bjorklundRandomizedSubexponentialAlgorithms2004, halmanSimpleStochasticGames2007}, where $|G|$ is the encoding size of the game. However, to the best of our knowledge, the best deterministic algorithms run in exponential time. As mentioned before, we can reduce such a game $G$ on an $n$-vertex graph to an instance of the $\ell_\infty$-contraction problem, yielding a $\lambda$-contracting function $f : [0, 1]^d \rightarrow [0, 1]^d$ with $d = n$, $\epsilon \approx 2^{-\poly(|G|)}$, and $\lambda \approx 1 - 2^{-\poly(|G|)}$~\cite{condonComplexityStochasticGames1992, etessamiTarskiTheoremSupermodular2020}. Thus, \Cref{thm:main_result_2} implies a deterministic algorithm for Condon's simple stochastic games running in time $|G|^{\bigO(\sqrt{n} \log n)}$.

Shapley's stochastic games are a generalization of Condon's simple stochastic games and are also played on a directed graph with $n$ vertices. Since exact solutions may be irrational, one often considers the problem of solving them approximately. The current state-of-the-art algorithm\footnote{See also \url{https://theorydish.blog/2025/09/15/quickly-approximating-shapley-games/} for a discussion.} for $\epsilon$-approximating the value of a Shapley stochastic game $G$ follows by applying the algorithm for monotone contractions from~\cite{batziouMonotoneContractions2025} which runs in time $\bigO((c \log \frac{1}{\epsilon})^{\lceil d / 3 \rceil})$: Indeed, by combining their algorithm with the reduction in~\cite{etessamiTarskiTheoremSupermodular2020}, they can $\epsilon$-approximately solve a Shapley stochastic game $G$ on a graph with $n$ vertices in time $(c \log \frac{1}{\epsilon})^{\lceil n / 3 \rceil}\poly(|G|)$ for some constant $c$, where $|G|$ denotes the encoding size of the game. By replacing the algorithm for monotone contractions from~\cite{batziouMonotoneContractions2025} with our new algorithm in \Cref{thm:main_result_2}, we thus obtain an algorithm for $\epsilon$-approximately solving a Shapley stochastic game in time $(\log \frac{1}{\epsilon})^{\bigO(\sqrt{n} \log n)}\poly(|G|)$. 

\paragraph*{Further Related Work}

Naturally, the problem of finding (approximate) fixed points of contraction maps has also been studied for metrics other than the $\ell_\infty$-metric. For example, polynomial-time algorithms have been known in the Euclidean case ($\ell_2$-metric) for a while now~\cite{sikorskiComputationalComplexityFixed2009, huangApproximatingFixedPoints1999, sikorskiEllipsoidAlgorithmComputation1993}. 
Moreover, in the case of piecewise linear $\ell_p$-contractions with $p \in \N \cup \{\infty\}$, an algorithm running in time $\bigO((\log \frac{1}{\epsilon})^d)$ was given in~\cite{fearnleyUniqueEndPotential2020}. The case of $\ell_p$-metrics for general $p \in [1, \infty]$ also got more attention recently when, inspired by the query-efficient algorithm for $\ell_\infty$-contractions in~\cite{chenComputingFixedPoint2025}, it was proven in~\cite{haslebacherQueryEfficientFixpointsLpContractions2025} that an $\epsilon$-approximate fixed point of a function $f : [0, 1]^d \rightarrow [0, 1]^d$ can be found in $\poly(d, \log \frac{1}{\epsilon}, \log \frac{1}{1 - \lambda})$ queries as long as $f$ is $\lambda$-contracting for some $\ell_p$-metric. Of course, the catch of this result is again that it is unclear how to actually find the query points in a time-efficient manner. 

If the metric is not specified ahead of time but is part of the input, then the problem of finding an (approximate) fixed point is known to be CLS-complete~\cite{daskalakisConverseBanachFixed2018}. In fact, fixed point problems on contraction maps were among the initial motivating problems for CLS~\cite{daskalakisContinuousLocalSearch2011}. 

\paragraph*{Future Work}

Our algorithm in \Cref{thm:main_result_2} is enabled by two main ingredients: 

\begin{itemize}
    \item[1)] The decomposition theorem in~\cite{chenQuadraticSpeedupComputing2026} and the observation that such a decomposition theorem is useful in combination with algorithms that use few queries but a lot of time.
    \item[2)] A query-efficient algorithm with a time dependency that is bad only in terms of $d$ (in particular, having only a logarithmic dependency on $1 / \epsilon$ in the runtime of our algorithm in \Cref{thm:main_result_1} is essential).
\end{itemize}

Decomposition theorems similar to the one in~\cite{chenQuadraticSpeedupComputing2026} also exist for some other problems\footnote{For $\ell_1$-contraction, there is a closely related result in~\cite{chenQuadraticSpeedupComputing2026}.} such as Tarski~\cite{fearnleyFasterAlgorithmFinding2022} and Discrete Monotone Approximate Contraction~\cite{batziouMonotoneContractions2025}. Therefore, if one could also prove something analogous to ingredient 2) above for any of those problems, it would probably also imply an overall faster algorithm for the respective problem.

We also want to note that the bottleneck of our algorithm in \Cref{thm:main_result_1} lies in the volume computations (see \Cref{sec:volume_computation}) that we need to perform. Thus, by improving this part of the algorithm, one would likely be able to get both a better runtime bound in \Cref{thm:main_result_1} as well as a better overall algorithm in the sense of \Cref{thm:main_result_2}.

\section{Technical Overview}
\label{sec:technical_overview}

We use this section to give an overview of our arguments. The goal is to explain how \Cref{thm:main_result_1} and \Cref{thm:main_result_2} are obtained by combining various pieces together in the right way. Detailed proofs of the individual pieces are deferred to Sections \ref{sec:finding_centerpoint}, \ref{sec:volume_computation}, and \ref{sec:decomposition}.

\subsection{Preliminaries}
\label{ssec:preliminaries}

In terms of notation, we will use $[n] = \{1, \dots, n\}$ for any natural number $n$, $\|\cdot \|$ for the $\ell_\infty$-norm unless otherwise indicated (i.e.\ $\|x\| = \max_{i \in [d]} |x_i|$ for $x \in \R^d$), and $S^{d - 1}$ for the unit sphere in $\R^d$. All our subsets of $\R^d$ are measurable, and we will use $\vol (X)$ to denote the Lebesgue measure of $X \subseteq \R^d$. Given two points $x, y \in \R^d$, we use $x \leq y$ to say that $x_i \leq y_i$ for all $i \in [d]$. We repeat that a function $f : [0,1]^d \rightarrow [0, 1]^d$ is $\lambda$-contracting for some $0 \leq \lambda < 1$ if $\|f(x) - f(y)\| \leq \lambda \|x - y\|$ for all $x, y \in [0, 1]^d$. If this property holds for $\lambda = 1$, we call the function non-expansive. A point $x \in [0, 1]^d$ is a fixed point of $f$ if it satisfies $x = f(x)$. It is an $\epsilon$-approximate fixed point of $f$ if it satisfies $\|x - f(x)\| \leq \epsilon$.

\subsection{The Proof of \Cref{thm:main_result_1}}
\label{ssec:proof_thm_1}

We start by outlining the algorithm for \Cref{thm:main_result_1} (see also~\Cref{alg:algo_1}). In particular, we want to find an $\epsilon$-approximate fixed point of a $\lambda$-contracting function $f : [0, 1]^d \rightarrow [0, 1]^d$ in $\bigO(d^2 \log(1 / \epsilon))$ queries and $(\log \frac{1}{\epsilon})^{\bigO(d \log d)}$ time. 

\paragraph{$\ell_\infty$-Halfspaces and $\ell_\infty$-Centerpoints}

Our algorithm follows the ideas of the query-efficient algorithm in~\cite{chenComputingFixedPoint2025}, but we use some of the more recent terminology introduced in~\cite{haslebacherQueryEfficientFixpointsLpContractions2025}, such as $\ell_\infty$-halfspaces and $\ell_\infty$-centerpoints: Given $x \in \R^d$ and $v \in S^{d - 1}$, the $\ell_\infty$-halfspace $H_v(x)$ around $x$ in direction $v$ is defined as the set of points
\[
    H_v(x) = \{ y \in \R^d \mid \forall \delta > 0:\; \| y - x\| \leq \|y - (x - \delta v) \| \}
\]
that are at least as close (in the $\ell_\infty$-metric) to $x$ as to $x - \delta v$, for all $\delta > 0$. Observe that $H_v(x)$ is invariant under scaling $v$ with a positive scalar, and thus the constraint $v \in S^{d - 1}$ can equivalently be understood as $v \in \R^d \setminus \{0\}$ (we simply think of $v$ as a direction).

\begin{figure}[ht]
    \centering
    \includegraphics[width=0.5\linewidth]{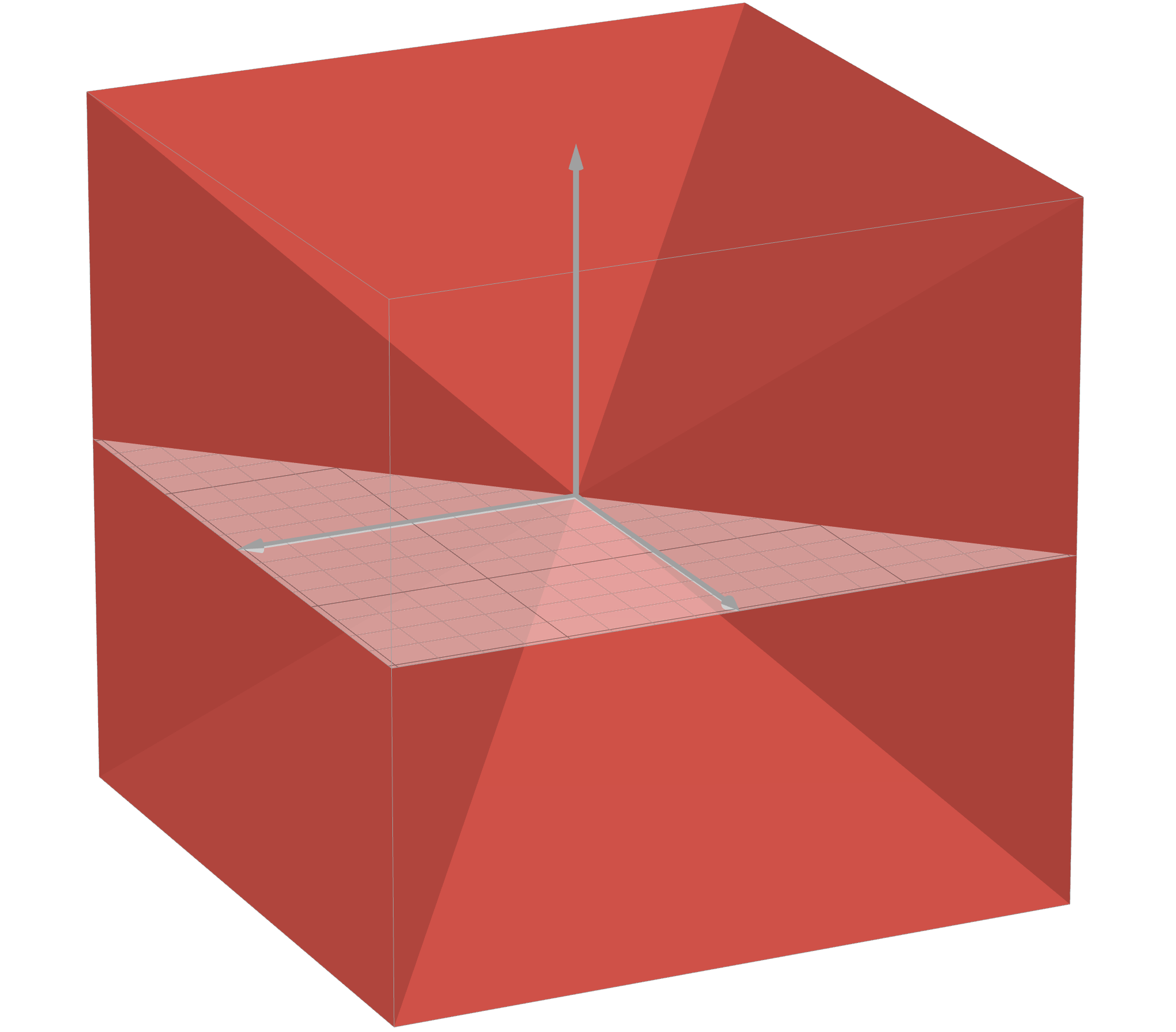}
    \caption{An $\ell_\infty$-halfspace $H_{v}(0)$ around the origin with $v = (-1, -1, -1)$ in $\R^3$ (coordinate axes drawn in gray). In particular, observe that $\ell_\infty$-halfspaces are in general not convex. Figure created with Desmos 3D calculator (\url{https://www.desmos.com/3d}).}
    \label{fig:halfspace}
\end{figure}

Next, consider a subset $X \subseteq [0, 1]^d$ with $\vol(X) \geq 0$. A point $c \in \R^d$ is an $\ell_\infty$-centerpoint of quality $\alpha$ of $X$ if and only if $\vol (X \cap H_v(c)) \geq \alpha \vol (X)$ for all $v \in S^{d - 1}$. We sometimes simply call such $c$ an $\alpha$-centerpoint of $X$. Crucially, the existence of $\frac{1}{2}$-centerpoints was proven in~\cite{chenComputingFixedPoint2025}\footnote{Technically, $X$ is a discrete pointset in their setting, but their proof also works in our continuous setting. See Appendix~\ref{appendix:reproduced_centerpoint_proof} for more details.}.

\begin{restatable}[$\ell_\infty$-Centerpoint Theorem~\cite{chenComputingFixedPoint2025}]{theorem}{thmcenterpoint}
\label{thm:centerpoint_theorem}
    For every $X \subseteq [0, 1]^d$ with $\vol(X) \geq 0$, there is a $\frac{1}{2}$-centerpoint in $[0, 1]^d$, i.e.\ a point $c \in [0, 1]^d$ with $\vol (X \cap H_v(c)) \geq \frac{1}{2} \vol (X)$ for all $v \in S^{d - 1}$.
\end{restatable}
\begin{proof}
    For convenience, we sketch the proof from~\cite{chenComputingFixedPoint2025} in Appendix~\ref{appendix:reproduced_centerpoint_proof}.
\end{proof}

\paragraph{Query-Efficient Algorithm}

We will now describe how the existence of $\ell_\infty$-centerpoints leads to a query-efficient algorithm for the $\ell_\infty$-contraction problem, following the ideas in~\cite{chenComputingFixedPoint2025}. For this, fix a $\lambda$-contracting function $f : [0, 1]^d \rightarrow [0, 1]^d$ with $\lambda < 1$. We start by initializing our search space $X = [0, 1]^d$. We then proceed in rounds: In each round, we query an $\alpha$-centerpoint $c \in [0, 1]^d$ of $X \subseteq [0, 1]^d$. By \Cref{thm:centerpoint_theorem}, such $c$ exists if $\alpha$ is chosen to be at most $1 / 2$. If $c$ is an $\epsilon$-approximate fixed point of $f$, we are done. Otherwise, we update $X \gets X \setminus H_v(c)$, where $v := c - f(c) \neq 0$, and continue with the next round. In particular, the algorithm only terminates once we have found an $\epsilon$-approximate fixed point of $f$.

\begin{algorithm}
\label{algo_thm_1}
    \caption{Finding an $\epsilon$-approximate fixed point of $f : [0, 1]^d \rightarrow [0, 1]^d$}
    \label{alg:algo_1}
    \DontPrintSemicolon

    $X \gets [0, 1]^d$ \\
    \Repeat{$\|c - f(c)\| \leq \epsilon$}{
        Compute $\alpha$-centerpoint $c$ of $X$ with \Cref{thm:approximate_centerpoint_theorem}, using \Cref{lemma:volume_computation} as a subroutine. \\
        $X \gets X \setminus H_{c - f(c)}(c)$
    }
    \Return $c$
\end{algorithm}

\paragraph{Correctness and Query Complexity}

To see why the above algorithm works, we adopt the analysis from~\cite{haslebacherQueryEfficientFixpointsLpContractions2025}. In particular, the algorithm terminates due to the following key lemma\footnote{The lemma does not appear exactly as written here in~\cite{haslebacherQueryEfficientFixpointsLpContractions2025}, which is why we provide a proof in Appendix~\ref{app:reproving}.}, which guarantees that as long as $c$ is not an $\epsilon$-approximate fixed point, the update $X \gets X \setminus H_{c - f(c)}(c)$ is safe.

\begin{restatable}[\cite{haslebacherQueryEfficientFixpointsLpContractions2025}]{lemma}{lemmaballcontained}
\label{lemma:ball_contained}
    Let $f : [0, 1]^d \rightarrow [0, 1]^d$ be $\lambda$-contracting with fixed point $x^\star = f(x^\star)$. Assume that $c \in [0, 1]^d$ is not an $\epsilon$-approximate fixed point of $f$. Then $B^\infty_r(x^\star) \cap H_{c - f(c)}(c) = \emptyset$, where $B^\infty_r(x^\star)$ is the $\ell_\infty$-norm ball of radius $r = \frac{\epsilon - \epsilon \lambda}{2+2\lambda}$ around $x^\star$. 
\end{restatable}

\begin{proof}
    For convenience, we reproduce the proof of~\cite{haslebacherQueryEfficientFixpointsLpContractions2025} in Appendix \ref{app:reproving}. 
\end{proof}

Note that the volume $\vol(B^\infty_r(x^\star) \cap [0, 1]^d)$ of $B^\infty_r(x^\star)$ from \Cref{lemma:ball_contained} intersected with $[0, 1]^d$ is at least $r^d$ (in the worst case $x^\star$ is at a corner of the cube), and the initial volume of $X$ is $1$ (since we initialize $X = [0, 1]^d$). As long as we do not query an $\epsilon$-approximate fixed point, the volume of our search space $X$ will shrink by at least a factor $(1 - \alpha)$ with every iteration (because $c$ is an $\alpha$-centerpoint of $X$, which implies $\vol(H_{c - f(c)}(c) \cap X) \geq \alpha \vol(X)$). Thus, \Cref{lemma:ball_contained} guarantees that the algorithm will have to query an $\epsilon$-approximate fixed point of $f$ within the first $\bigO(\frac{d}{\alpha} \log \frac{1}{r})$ iterations, since 
\[
    (1 - \alpha)^{\frac{d}{\alpha} \log \frac{1}{r}} \leq \exp ( -d \log (\frac{1}{r})) = r^d.
\]
This argument establishes correctness and an upper bound on the number of queries made by \Cref{alg:algo_1}. However, the key challenge will be to actually find the centerpoint $c$ of $X$ in each iteration in reasonable time. 

\paragraph*{Pyramids and Pyramid Decomposition}

As mentioned in \Cref{sec:introduction}, a brute-force search was used in~\cite{chenComputingFixedPoint2025} to find the centerpoints that are guaranteed to exist by \Cref{thm:centerpoint_theorem}. This resulted in a runtime of at least $\Omega((1 / \epsilon)^d)$. Our main technical contribution is to replace the brute-force search with a new procedure that can find a $\frac{1}{4d}$-centerpoint in time $(\log \frac{1}{\epsilon})^{\bigO(d \log d)}$. For this, we will heavily rely on the notion of pyramids, which were also already introduced in~\cite{chenComputingFixedPoint2025}: Given $x \in \R^d$, $i \in [d]$, and a sign $s \in \{-1, +1\}$, the $s$-pyramid $P^s_i(x)$ (see also \Cref{fig:pyramid}) in direction $i$ around $x$ is defined as 
\[
    P^s_i(x) = \{ y \in \R^d \mid s(y_i - x_i) = \|y - x\| \}.
\]
We also sometimes refer to $P^-_i(x)$ ($P^+_i(c)$, respectively) as the negative (positive) pyramid in direction $i$ around $x$. Note that there are exactly $2d$ different pyramids around every point $x \in \R^d$ (one negative and one positive pyramid for each $i \in [d]$). If the parameters $i, s, x$ are not important for an argument, we sometimes refer to a pyramid $P$ without specifying them.

\begin{figure}[ht]
    \centering
    \includegraphics[width=0.5\linewidth]{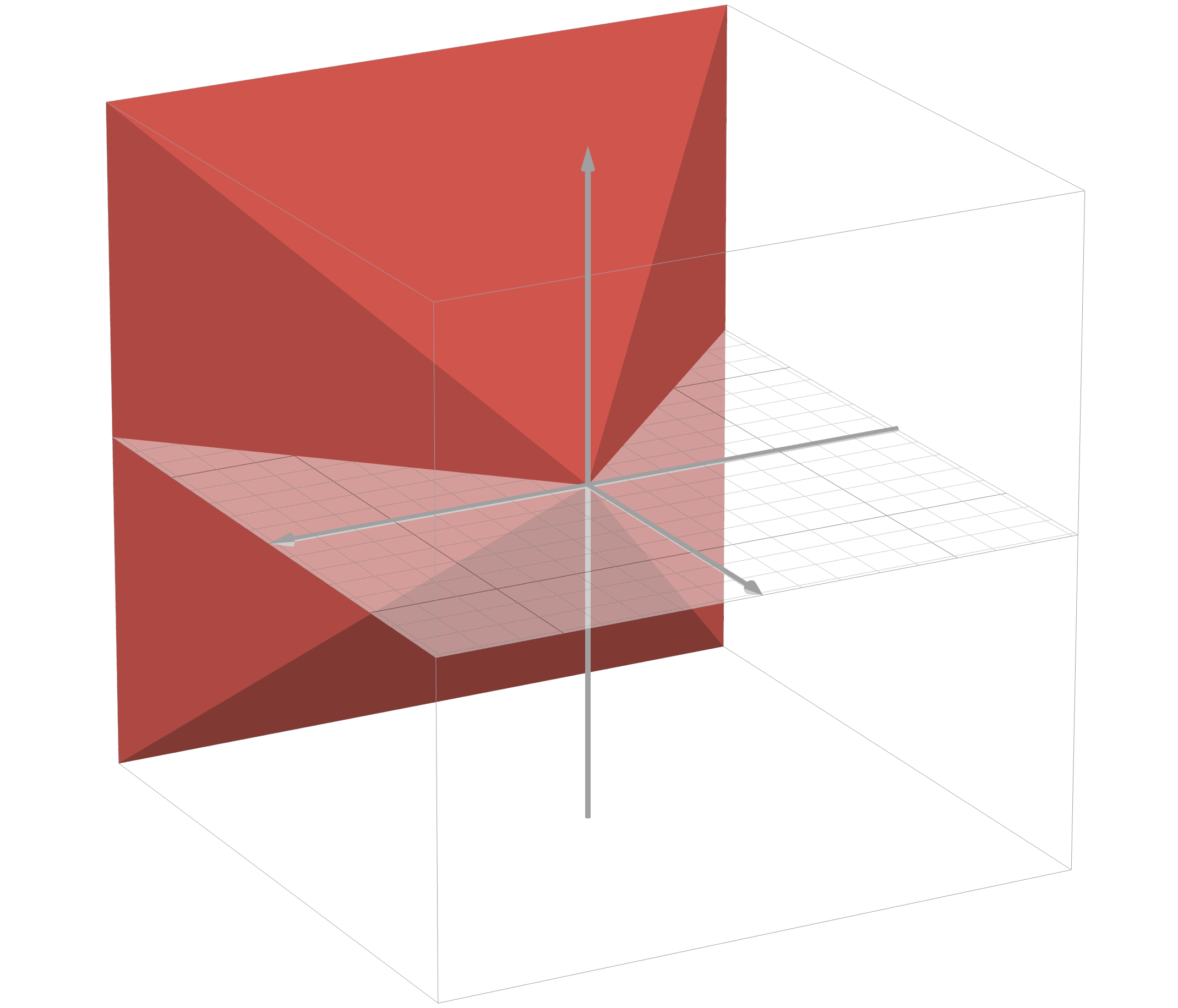}
    \caption{A pyramid around the origin in $\R^3$. Observe that all pyramids are convex. Figure created with Desmos 3D calculator (\url{https://www.desmos.com/3d}).}
    \label{fig:pyramid}
\end{figure}

Pyramids arise naturally due to the observation in~\cite{chenComputingFixedPoint2025} that every $\ell_\infty$-halfspace can be decomposed into pyramids (for example, the $\ell_\infty$-halfspace in \Cref{fig:halfspace} consists of the three negative pyramids around the origin), as follows.

\begin{observation}[Pyramid Decomposition~\cite{chenComputingFixedPoint2025}]
\label{obs:pyramid_decomposition}
    For all $x \in \R^d$ and $v \in S^{d - 1}$, the $\ell_\infty$-halfspace $H_v(x)$ can be decomposed into pyramids as
    \[
        H_v(x) = \left( \bigcup_{i : v_i \geq 0} P^+_i(x) \right) \cup \left( \bigcup_{i : v_i \leq 0} P^-_i(x) \right).
    \]
\end{observation}

In particular, we have 
\[
    \vol(X \cap H_v(x)) = \left( \sum_{i : v_i \geq 0} \vol(P^+_i(x) \cap X) \right) + \left( \sum_{i : v_i \leq 0} \vol(P^-_i(x) \cap X) \right)
\]
for all $x \in \R^d$ and $X \subseteq [0, 1]^d$. Moreover, \Cref{obs:pyramid_decomposition} tells us that we could equivalently restrict $v$ to the domain $\{-1, 0, +1\}^d \setminus \{0\}$ instead of $S^{d - 1}$ or $\R^d \setminus \{0\}$ in the definition of $\ell_\infty$-halfspaces. Indeed, for every direction $i \in [d]$, $H_v(x)$ contains $P^-_i(x)$ or $P^+_i(x)$, depending only on the sign of $v_i$. This also implies that there are exactly $3^d - 1$ different $\ell_\infty$-halfspaces around any given point $x \in \R^d$.

\paragraph{Valid Search Space} 
With \Cref{obs:pyramid_decomposition}, we can now analyze the shape of our search space $X \subseteq [0, 1]^d$ at each step of \Cref{alg:algo_1}. In particular, recall that we start with $X = [0, 1]^d$ and repeatedly cut away $\ell_\infty$-halfspaces from $X$. By \Cref{obs:pyramid_decomposition}, we can also reformulate this as starting with $X = [0, 1]^d$ and repeatedly cutting away pyramids from it, since each $\ell_\infty$-halfspace is a union of at most $2d - 1$ pyramids. This leads us to our definition of valid search spaces. 

\begin{definition}[Valid Search Space]
    A set $X \subseteq [0, 1]^d$ is a valid search space of complexity $n$ if and only if there exist $n$ pyramids $P^{(1)}, \dots, P^{(n)}$ such that $X = [0, 1]^d \setminus \bigcup_{i \in [n]} P^{(i)}$.
\end{definition}

\begin{observation}
\label{obs:search_space_complexity}
    Throughout \Cref{alg:algo_1}, the set $X$ is always a valid search space of complexity at most $\bigO \big( d q \big)$, where $q$ is an upper bound on the number of iterations (or equivalently queries) made by the algorithm.
\end{observation}

The definition of valid search spaces is useful to us, because it turns out that they allow for the following volume computation procedure.

\begin{restatable}[Volume Computation]{lemma}{lemmavolumecomputation}
\label{lemma:volume_computation}
    Consider an arbitrary valid search space $X$ of complexity $n$ and an arbitrary pyramid $P$ in $\R^d$. The volume $\vol(X \cap P)$ can be computed in $(nd)^{\bigO(d)}$ time. 
\end{restatable}

The proof of \Cref{lemma:volume_computation} heavily relies on results from computational geometry and will be explained in detail in \Cref{sec:volume_computation}. 

\paragraph{Finding Centerpoints and \Cref{thm:main_result_1}} 

With all these prerequisites in place, we can now state our main technical contribution, which is an algorithm for finding a $\frac{1}{4d}$-centerpoint of a valid search space. The proof is given in \Cref{sec:finding_centerpoint}, and in particular, \Cref{subsec:centerpoint_intuition} contains an informal overview of the ideas. 

\begin{restatable}[Finding Approximate Centerpoints]{theorem}{thmapproxcenterpoint}
\label{thm:approximate_centerpoint_theorem}
    Let $X \subseteq [0, 1]^d$ be a valid search space with $\vol(X) > 0$. There is an algorithm that finds a $\frac{1}{4d}$-centerpoint of $X$ in time $\bigO(d^7T\log^2 \frac{d}{\vol X})$, assuming it has access to an oracle that computes $\vol(X \cap P)$ for any pyramid $P$ in time $T$.
\end{restatable}

For our application, the runtime of the algorithm in \Cref{thm:approximate_centerpoint_theorem} will be dominated by $T$: Indeed, we use \Cref{lemma:volume_computation} for the volume computation oracle, which has an exponential dependency on $d$. Still, this now allows us to prove \Cref{thm:main_result_1}.

\thmmainresultone*

\begin{proof}
    We start by using an observation from~\cite{chenComputingFixedPoint2025}, which says there is a simple and efficient reduction that allows us to assume $\lambda \leq 1 - \epsilon$ (see also \Cref{ssec:reductions} for more information on this reduction). With this, we now apply \Cref{alg:algo_1}, which repeatedly calls the algorithm in \Cref{thm:approximate_centerpoint_theorem} to find the centerpoints (which are our query points). Since \Cref{thm:approximate_centerpoint_theorem} guarantees $\frac{1}{4d}$-centerpoints, \Cref{alg:algo_1} terminates after at most $\bigO(d^2\log 1 / r) \leq \bigO(d^2 \log \frac{1}{\epsilon})$ queries to $f$ (where we have used \Cref{lemma:ball_contained} with $r = \frac{\epsilon - \epsilon \lambda}{2 + 2\lambda}$ and $\lambda \leq 1 - \epsilon$). By \Cref{obs:search_space_complexity}, this means that throughout \Cref{alg:algo_1}, the complexity of our valid search space is bounded by $\bigO(d^3 \log \frac{1}{\epsilon})$.
    
    Next, note that we only call \Cref{thm:approximate_centerpoint_theorem} for $X$ with $\vol(X) \geq r^d$ (due to \Cref{lemma:ball_contained}). Thus, the runtime of \Cref{thm:approximate_centerpoint_theorem} is $\poly(d, \log \frac{1}{\epsilon})$ except for the volume computation, which is handled by the subroutine in \Cref{lemma:volume_computation}. Since the complexity of our valid search space is bounded throughout by $\bigO(d^3 \log \frac{1}{\epsilon})$, we get an upper bound of $(\log \frac{1}{\epsilon})^{\bigO(d\log d)}$ for each application of \Cref{lemma:volume_computation}. This of course dominates the runtime of \Cref{thm:approximate_centerpoint_theorem} and the overall algorithm.
\end{proof}

\paragraph{Bit Complexity and Precision}

In all of our algorithms, we count runtime in terms of the number of arithmetic operations. Naturally, one might also wonder about the bit complexity of the involved numbers. It turns out that all of our points $x \in \R^d$ that we compute throughout the algorithm have polynomially bounded bit complexity (i.e.\ at most $\poly(d, \log \frac{1}{\epsilon})$ bits). This follows from the use of the precision lemma below, which ensures that we only need to move on a grid of limited granularity when looking for an approximate centerpoint $c$: Indeed, slightly moving $c$ can only slightly change the volume $\vol(X \cap H_v(c))$ for any $v \in S^{d - 1}$. 

\begin{lemma}[Precision Lemma]
\label{lemma:precision}
    Let $X \subseteq [0, 1]^d$ be arbitrary with volume $\vol(X) \geq 0$ and consider an arbitrary pyramid $P^s_i(y)$ for some $y \in \R^d$. Then we have 
    \[
        \left| \vol(X \cap P^s_i(y)) - \vol(X \cap P^s_i(y + u)) \right| \leq \sqrt{8d^3} \|u\|
    \]
    for all $u \in \R^d$.
\end{lemma}
\begin{proof}
    We use a well-known result due to Ball~\cite{ballCubeSlicingRn1986}, who proved that the largest $(d - 1)$-dimensional volume that one can get by slicing the unit cube $[0, 1]^d \subseteq \R^d$ with a hyperplane is $\sqrt{2}$. Thus, slicing the cube with a slab of parallel hyperplanes of thickness at most $\sqrt{\sum_{i = 1}^d u_i^2} \leq \sqrt{d} \|u\|$ yields a $d$-dimensional volume of at most $\sqrt{2d} \|u\|$. Now it remains to observe that any pyramid is bounded by $2 (d - 1)$ hyperplanes (this fact will be explained in more detail in \Cref{sec:volume_computation}), which means that the change of volume from $\vol(X \cap P^s_i(y))$ to $\vol(X \cap P^s_i(y + u))$ is bounded by $\sqrt{8d^3} \|u\|$.
\end{proof}

The one place where we have to be a bit more careful about bit complexity of numbers is the volume computations (\Cref{lemma:volume_computation}). Indeed, it turns out that a number of bits exponential in $d$ might be necessary to represent the computed volume (more details in \Cref{sec:volume_computation}). However, since the stated runtime of \Cref{lemma:volume_computation} is also exponential in $d$, this would only be a problem if we then propagate this blow-up in bit complexity through \Cref{alg:algo_1}. But since we only ever use computed volumes for comparison (see \Cref{sec:finding_centerpoint}), this is not the case. 

Finally, observe that \Cref{alg:algo_1} uses the result $f(c)$ of querying $f$ at $c$ to check whether $c$ is an $\epsilon$-approximate fixed point, and to compute $v = c - f(c)$ in order to determine the halfspace $H_v(c)$ that we want to exclude from $X$. Thus, when also accounting for the cost of arithmetic operations, we (naturally) need to assume\footnote{This is true for the functions that are obtained by the reductions from Condon's and Shapley's stochastic games.} that for all $x \in \Q^d \cap [0, 1]^d$, the bit complexity of $f(x)$ is polynomially bounded in the bit complexity of $x$.

\subsection{Proof of \Cref{thm:main_result_2}}
\label{ssec:proof_thm_2}

As mentioned in \Cref{sec:introduction}, we obtain \Cref{thm:main_result_2} by combining \Cref{thm:main_result_1} with a recent decomposition theorem in~\cite{chenQuadraticSpeedupComputing2026}. We will briefly outline the main ideas for this here, and provide the remaining details in \Cref{sec:decomposition}. We follow the exposition in~\cite{chenQuadraticSpeedupComputing2026} and define the following terminology:

\begin{itemize}
    \item $\contr(d, \epsilon, \lambda)$ denotes the problem of finding an $\epsilon$-approximate fixed point of a $\lambda$-contracting function $f : [0, 1]^d \rightarrow [0, 1]^d$. 
    \item $\nonexp^\dagger(d, \epsilon)$ denotes the problem of finding an $\epsilon$-approximate fixed point of a non-expansive function $f : [0, 1]^d \rightarrow [- \epsilon, 1 + \epsilon|^d$. Note that this is a well-defined problem (i.e.\ such an $\epsilon$-approximate fixed point must exist), as argued in~\cite{chenQuadraticSpeedupComputing2026}.
\end{itemize}

Arguably the main result in~\cite{chenQuadraticSpeedupComputing2026} is \Cref{thm:decomposition_theorem} below. 

\begin{restatable}[Decomposition Theorem~\cite{chenQuadraticSpeedupComputing2026}]{theorem}{thmdecompositiontheorem}
    \label{thm:decomposition_theorem}
    Assume that $\nonexp^\dagger(d_1, \epsilon)$ can be solved in $q_1$ queries and $t_1$ time, and assume that $\nonexp^\dagger(d_2, \epsilon)$ can be solved in $q_2$ queries and $t_2$ time, then $\nonexp^\dagger(d_1 + d_2, \epsilon)$ can be solved in $q_1 \cdot q_2$ queries and $\bigO((q_2 + t_1)q_2 + t_2)$ time.
\end{restatable}

Note that the original statement in~\cite{chenQuadraticSpeedupComputing2026} states an upper bound of $\bigO((q_2 + t_1)t_2)$ for the time needed to solve $\nonexp^\dagger(d_1 + d_2, \epsilon)$, which is more concise in the case $q_2 \approx t_2$. However, it is clear from their analysis\footnote{The analysis of Algorithm 1 in~\cite{chenQuadraticSpeedupComputing2026} can be found at the end of page 9.} that the bound that we stated holds as well (and in particular, the multiplicative term $t_1 t_2$ can be avoided): Indeed, the inner algorithm (solving $\nonexp^\dagger(d_1, \epsilon)$) is called once for every query of the outer algorithm (of which there are $q_2$ many). The additional time $t_2$ spent by the outer algorithm has no interaction with the inner algorithm, and thus it suffices to make it an additive term in the upper bound $\bigO((q_2 + t_1)q_2 + t_2)$.

Avoiding the multiplicative term $t_1 t_2$ is what enables us to use \Cref{thm:decomposition_theorem} in order to trade off query complexity against time complexity. By further combining this with a reduction from $\nonexp^\dagger(d, \epsilon)$ to $\contr(d, \epsilon / 4, 1 - \epsilon / 4)$, we prove the following lemma in \Cref{sec:decomposition}.

\begin{restatable}{lemma}{squarerootlemma}
\label{lemma:square-root-decomposition}
    Assume that $\contr(d, \epsilon, 1 - \epsilon)$ can be solved in $q$ queries and $t$ time. Then $\nonexp^{\dagger}(d^2, 4\epsilon)$ can be solved in $(q + 1)^d$ queries and at most $q^{\bigO(d)} t$ time.
\end{restatable}

In combination with \Cref{thm:main_result_1} and with the straight-forward reduction from $\contr(d, \epsilon, \gamma)$ to $\nonexp^\dagger(d, \epsilon)$, this then implies \Cref{thm:main_result_2}.

\thmmainresulttwo*

\section{Finding Approximate Centerpoints}
\label{sec:finding_centerpoint}

The goal of this section is to prove \Cref{thm:approximate_centerpoint_theorem}, which we restate below for convenience. Since this is our main technical contribution, we proceed by first giving an overview over the involved ideas (\Cref{subsec:centerpoint_intuition}) before discussing some preliminaries (\Cref{ssec:pulling}) and then presenting the algorithms and proofs in detail (\Cref{subsec:centerpoint_details}).

\thmapproxcenterpoint*

\subsection{Overview of Ideas}
\label{subsec:centerpoint_intuition}

\paragraph*{Balancing Pyramids with Local Search}
First, observe that a point $c \in \R^d$ which satisfies $\vol(X \cap P^+_i(c)) \approx \vol(X \cap P^-_i(c))$ for all $i \in [d]$ is a good approximate $\ell_\infty$-centerpoint of $X$ by \Cref{obs:pyramid_decomposition} (pyramid decomposition of $\ell_\infty$-halfspaces). This motivates the following approach: Consider the pairs of opposing pyramids around an arbitrary point $c$. Looking at each pair, it is either ``balanced'' (both the negative and positive pyramid contain roughly the same amount of volume of $X$), or one of the two pyramids is heavier than the other. Thus, one might be tempted to move $c$ in the direction of the heavy pyramids in order to hopefully decrease their volume while at the same time increasing the volume of the light pyramids. 

The main issue with this approach is what happens to the already balanced pairs of opposing pyramids. Indeed, if we have $\vol(X \cap P^+_i(c)) \approx \vol(X \cap P^-_i(c))$ for some $i \in [d]$, the above approach suggests that we should not move in direction $i$ at all. But then, while performing moves for other pairs of pyramids, $\vol(X \cap P^+_i(c))$ and $\vol(X \cap P^-_i(c))$ could become unbalanced again. In particular, they can independently increase and decrease in volume (see e.g.\ \Cref{fig:pulling}), and it seems difficult to control this.

However, we can still use this approach in a situation where pyramid pairs of every dimension $i  \in [d]$ are unbalanced. More concretely, consider $c \in \R^d$ such that $\vol(P^-_i(c) \cap X) \geq m$ for all $i \in [d]$ and some $m \geq 0$. Then we can move $c$ in the direction $-\One = -(1, \dots, 1) \in \R^d$ until roughly $\frac{m}{2}$ of the volume of $X$ has moved from negative pyramids around $c$ to positive pyramids around $c$ (see \Cref{corollary:regions_to_pyramids}). It turns out that this is then sufficient to roughly ensure $\vol(X \cap H_v(c)) \geq \frac{m}{2}$ for every $\ell_\infty$-halfspace $H_v(c)$ around $c$ (see \Cref{obs:m/2 centerpoint}). Overall, this argument reduces our problem to finding an initial point $c$ such that $\vol(X)$ is evenly distributed among the negative pyramids around $c$ (with each e.g.\ containing roughly $m \approx \vol(X) / d$ volume of $X$). 

\paragraph*{Balancing Negative Pyramids}
Consider now this task of balancing the negative pyramids around our candidate $c \in \R^d$. To explain how we achieve this, assume first that a unicorn magically guarantees that all positive pyramids around $c$ will have zero volume (w.r.t.\ $X$) throughout this part of the algorithm. This guarantee ensures us that, while moving $c$ around, any volume that a negative pyramid $P^-_i(c)$ around $c$ loses (w.r.t\ to $X$) will be lost to different negative pyramids around $c$. 

Interestingly, with the unicorn's guarantee, we can now recover a local search approach again: In each iteration, we consider a set of negative pyramids with substantially more volume than the rest and move $c$ in the direction of this set of negative pyramids (see also sketch in \Cref{fig:balancing_neg}). It turns out that we can perform such a step in a way that guarantees that the negative pyramids overall become more balanced (using \Cref{lemma:pulling} and \Cref{lemma:pulling_operation}).
In \Cref{lemma:balancing_pyramids}, we then analyze how many such steps we need to perform until we reach a reasonably balanced configuration.

Unfortunately, unicorns (probably) do not exist in real life. However, it turns out that we can ensure the guarantee of the unicorn ourselves as follows: We will initialize $c = (1, 1, \dots, 1) \in \R^d$, which ensures $\vol(P^+_i(c) \cap [0,1]^d) = 0$ for all $i \in [d]$ (by the definition of pyramids). While moving $c$ in our iterative balancing, we make sure that no coordinate of $c$ ever decreases. In particular, we will maintain $1 \leq c_i$ for all $i \in [d]$ and thus $\vol(P^+_i(c) \cap [0,1]^d) = 0$ for all $i \in [d]$ throughout this part of the algorithm.

\paragraph*{Projecting to the Cube}
We outlined our strategy of first balancing the negative pyramids around $c$, and then moving along the line in direction $-\One$ through $c$ to find an approximate centerpoint of $X$. However, the resulting approximate centerpoint might still lie outside of $[0, 1]^d$, which is a problem if we want to query the function $f$ at $c$ in \Cref{alg:algo_1}. To solve this, we prove in \Cref{lemma:centerpoint_projection_on_cube} that by projecting $c$ to $[0, 1]^d$ in the right way, we can keep some of the guarantees on the volume of its pyramids (w.r.t.\ $X$), ensuring that the projected point is again an approximate centerpoint. 

\subsection{The Pulling Operation}
\label{ssec:pulling}

Before moving on to the proof of \Cref{thm:approximate_centerpoint_theorem} and the involved algorithms, we make some preliminary observations. In \Cref{subsec:centerpoint_intuition}, we repeatedly used the notion of moving a candidate point $c \in \R^d$ around in order to balance the amount of volume of $X$ contained in different pyramids around $c$. Here, we formalize the notion of moving with an operation that we will refer to as \emph{pulling} $c$ with a certain strength $\alpha > 0$ in a direction $u \in \{-1, 0, +1\}^d$, which simply means that we move $c$ to the position $c + \alpha u$. We will also sometimes refer to the pulling-direction $u$ implicitly in terms of its associated pyramids: For example, we might say that we are pulling in the direction of the pyramids $P^-_i(c)$ and $P^+_j(c)$, which means that we are pulling in the direction $u \in \{-1, 0, +1\}^d$ with $u_i = -1$, $u_j = +1$, and $u_k = 0$ for all $k \notin \{i, j\}$. Naturally, we can never pull in the direction of $P^+_i(c)$ and $P^-_i(c)$ simultaneously, for any $i \in [d]$. In fact, in the above example, we might equivalently say that we are pulling $c$ away from $P^+_i(c)$ and $P^-_j(c)$.

As a preparation for later proofs, we will need to analyze how points can be lost and gained between different pyramids around $c$ throughout pulling operations (see for example \Cref{fig:pulling} for a sketch in $d = 2$). This is what \Cref{lemma:pulling} below achieves. Intuitively, it groups the pyramids around $c$ into three groups: Pyramids we pull in the direction of (A), pyramids we pull away from (B), and the remaining pyramids (C). The lemma says that points can only be lost from A to B, A to C, or C to B.

\begin{figure}[ht]
    \centering
    \includegraphics[width=0.3\linewidth]{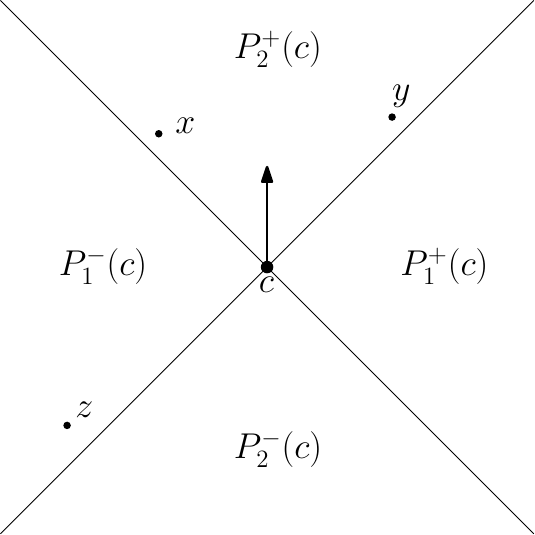}
    \caption{A simple example of \Cref{lemma:pulling} for $d = 2$. Observe how pulling $c$ in the direction $u = (0, 1)$ (indicated by the arrow) implies that $P^+_2(c)$ may lose points $x$ and $y$ to $P^-_1(c)$ and $P^+_1(c)$ respectively (or it could even lose them to $P^-_2(c)$ if we pull far enough). Similarly, $P^-_1(c)$ might lose point $z$ to $P^-_2$. However, $P^-_2$ will not lose any points.}
    \label{fig:pulling}
\end{figure}

\begin{lemma}
\label{lemma:pulling}
    Consider arbitrary $c \in \R^d$, $\alpha > 0$, and non-zero $u \in \{-1, 0, +1\}^d$. Consider an arbitrary point $x \in \R^d$ and an arbitrary pyramid $P^s_i(c)$ around $c$ with $i \in [d]$ and $s \in \{-1, +1\}$. Then the following two statements hold: 
    \begin{itemize}
        \item If $x \in P^s_i(c)$ but $x \notin P^s_i(c + \alpha u)$, then $s \cdot u_i \geq 0$.
        \item If $x \notin P^s_i(c)$ but $x \in P^s_i(c + \alpha u)$, then $s \cdot u_i \leq 0$.
        \item If $x \notin P^s_i(c)$ but $x \in P^s_i(c + \alpha u)$ and $u_i = 0$, then $x \in P^{s'}_j(c)$ for some $j \in [d]$ with $s' \cdot u_j > 0$.
    \end{itemize}
\end{lemma}
\begin{proof}
    We start by proving the first statement in an indirect fashion: To this end, assume $s \cdot u_i < 0$. Now observe that by the definition of pyramids, this implies $P^s_i(c) \subseteq P^s_i(c + \alpha u)$, since $u \in \{-1, 0, +1\}^d$ with $u_i = -s$.
    Consequently, $x \in P^s_i(c)$ would also imply $x \in P^s_i(c + \alpha u)$, and thus the first statement holds. The proof of the second statement is analogous. 
    
    For the third statement, assume $u_i = 0$, $x \notin P^s_i(c)$, and $x \in P^s_i(c + \alpha u) \cap P_j^{s'}(c)$ for some $j \in [d]$ (observe that $x$ must be in at least one of the pyramids around $c$). By the definition of the pyramids, we thus have
    \[
         s'(x_j -c_j) = \|x-c\| > s(x_i-c_i)
        =s(x_i-(c_i+0))=s(x_i-(c_i+\alpha u_i)),
    \]
    but also
    \[
        s'(x_j-(c_j+\alpha u_j))\leq \|x-(c+\alpha u)\| = s(x_i-(c_i+\alpha u_i)).
    \]
    Together, this implies $s'(x_j-c_j) > s'(x_j - (c_j+\alpha u_j))$, resulting in $s'\cdot u_j>0$ since $\alpha>0$.
\end{proof}

We will also need an upper bound on how far we have to pull. The following lemma provides this: Intuitively, it says that pulling $c$ with strength $\alpha = 2\|c\| + 2$ in any direction $u \in \{-1, 0, +1\}^d$ will guarantee that the pyramids that we pull away from will contain all of $[0, 1]^d$.

\begin{lemma}
\label{lemma:pulling_upper_bound}
    Consider arbitrary $c \in \R^d$ and non-zero $u \in \{-1, 0, +1\}^d$, and define $S = \{i \in [d] \mid u_i \neq 0\}$ as the set of indices where $u$ is non-zero. Then for $\alpha = 2\|c\| + 2$,
    we have 
    \[
        [0, 1]^d \subseteq \left( \bigcup_{i \in S} P_i^{-u_i}(c + \alpha u) \right).
    \]
\end{lemma}
\begin{proof}
    Let $x \in [0, 1]^d$ be arbitrary and consider an arbitrary pyramid $P_i^{-u_i}(c + \alpha u)$ with $i \in S$. Observe that we have
    \[
        -u_i(x_i - (c_i + \alpha u_i)) =  -u_i x_i + u_i c_i + \alpha u_i^2 \geq \|c\| + 1
    \]
    by our choice of $\alpha = 2\|c\| + 2$. Now consider arbitrary $j \notin S$. Since $j \notin S$, we must have $u_j = 0$ and thus
    \[
        |s(x_j - (c_j + \alpha u_j))| = |s(x_j - c_j)| \leq \|c\| + 1.
    \]
    This implies that $\|x - (c - \alpha u)\| = -u_i(x_i - (c_i + \alpha u_i))$ for some $i \in S$, and thus $x \in P_i^{-u_i}(c + \alpha u)$ for that $i$, which concludes the proof.
\end{proof}

\begin{figure}[ht]
    \centering
    \includegraphics[width=0.3\linewidth]{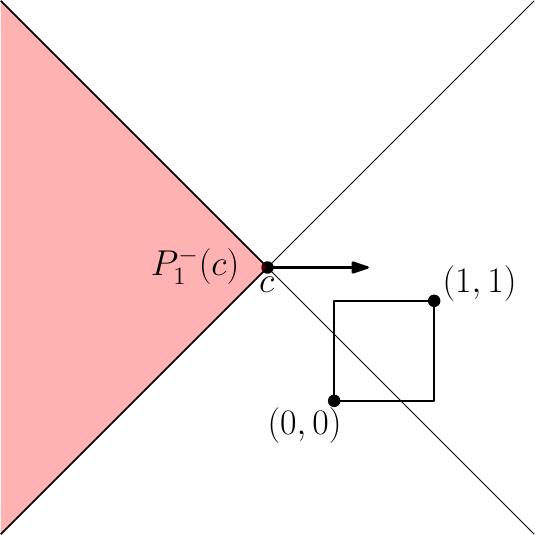}
    \caption{\Cref{lemma:pulling_upper_bound} gives a sufficient bound for how far we have to pull $c$ in the direction $(1, 0)$ until $P^-_1(c)$ contains $[0, 1]^2$.}
    \label{fig:max_pulling}
\end{figure}

With \Cref{lemma:pulling} and \Cref{lemma:pulling_upper_bound}, we are now ready to conclude \Cref{lemma:pulling_operation} below, which is the main tool that we will use from now on: \Cref{lemma:pulling_operation} gives a binary search algorithm for finding the appropriate strength with which we should pull in a given direction in order to achieve a certain volume increase in the set of pyramids that we pull away from.

\begin{lemma}[The Pulling Operation]
\label{lemma:pulling_operation}
    Let $X \subseteq [0, 1]^d$ with $\vol(X) \geq 0$, and $\epsilon > 0$. Consider $c \in \R^d$ and non-zero $u \in \{-1, 0, +1\}^d$, and define again $S = \{i \in [d] \mid u_i \neq 0\}$. Let $m$ in the range $\vol(X) \geq m \geq \vol \left( X \cap \bigcup_{i \in S} P_i^{-u_i}(c) \right)$ be arbitrary. Then we can find a point $c'$ satisfying 
    \[
        m \geq \vol \left( X \cap \bigcup_{i \in S} P_i^{-u_i}(c') \right) \geq m - \epsilon
    \]
    in time $\bigO(dT\log \frac{\|c\|d}{\epsilon})$, where we assume that we have access to an oracle that computes $\vol(X \cap P)$ for any pyramid $P$ in time $T$.
\end{lemma}
\begin{proof}
    Consider the function 
    \[
        F(\alpha) = \vol \left( X \cap \bigcup_{i \in S} P_i^{-u_i}(c + \alpha u) \right) = \sum_{i \in S} \vol(P^{-u_i}_i(c + \alpha u ) \cap X)
    \]
    for $\alpha \geq 0$. Observe that $F$ is non-decreasing by \Cref{lemma:pulling} and satisfies $F(\alpha) = \vol(X)$ for $\alpha \geq 2\|c\| + 2$ by \Cref{lemma:pulling_upper_bound}. Thus, we can use binary search in the range $\alpha \in [0, 2\|c\| + 2]$ to find two values $\alpha_\ell, \alpha_r$ with $0 \leq \alpha_\ell < \alpha_r < \alpha_\ell + \frac{\epsilon}{\sqrt{8d^5}}$ 
    and $f(\alpha_\ell) \leq m \leq f(\alpha_r)$. By \Cref{lemma:precision} (precision lemma), we get $f(\alpha_\ell) \geq m - \epsilon$ (since we consider at most $d$ pyramids, each of which can gain at most $\sqrt{8d^3}\frac{\epsilon}{\sqrt{8d^5}} = \frac{\epsilon}{d}$ volume between $\alpha_\ell$ and $\alpha_r$), which means that $c' = c + \alpha_\ell u$ is the point we are looking for. 
    
    Since our binary search starts on a domain of size $2\|c\| + 2$ and ends as soon as $\alpha_r - \alpha_\ell < \epsilon / \sqrt{8d^5}$, it runs in $\bigO(\log \frac{\|c\|d}{\epsilon})$ steps, with each step requiring $\bigO(dT)$ time for the volume computations. 
\end{proof}

As a first application of \Cref{lemma:pulling_operation}, let us state the following corollary, which explains how we can find a reasonably good approximate centerpoint once we have a point $c$ whose negative pyramids are sufficiently balanced. 

\begin{corollary}
\label{corollary:regions_to_pyramids}
    Let $X \subseteq [0,1]^d$ with volume $\vol(X) \geq 0$ and $\epsilon > 0$ be arbitrary. Assume that we are given $c \in \R^d$ with $\vol(P^+_i(c) \cap X) = 0$ and $\vol(P^-_i(c) \cap X) \geq m$ for all $i \in [d]$ and some $m \geq 0$. Then we can find a point $c' \in \R^d$ with $\sum_{i \in [d]} \vol(P^+_i(c') \cap X)\geq \frac{m}{2} - \epsilon$ and $\vol(P^-_i(c') \cap X) \geq \frac{m}{2}$ for all $i \in [d]$ in time $\bigO(dT\log \frac{\|c\|d}{\epsilon})$, where we assume that we have access to an oracle that computes $\vol(X \cap P)$ for any pyramid $P$ in time $T$.
\end{corollary}
\begin{proof}
    Consider the direction $u = -\One = -(1, \dots, 1)$ and apply \Cref{lemma:pulling_operation} to find a point $c' \in \R^d$ that satisfies  
    \[
        \frac{m}{2} \geq \vol \left( X \cap \bigcup_{i \in [d]} P_i^{+}(c') \right) \geq \frac{m}{2} - \epsilon.
    \]
    Observe that by \Cref{lemma:pulling}, negative pyramids can only lose and positive pyramids only gain volume while pulling $c$ in the direction of $u$. Thus, since all the negative pyramids together lost at most $\frac{m}{2}$ of volume, we still have $\vol(P^-_i(c') \cap X) \geq \frac{m}{2}$ for all $i \in [d]$ after the pulling (in the worst case, all of the volume was lost in the same negative pyramid).
\end{proof}

We already said that the point $c'$ in \Cref{corollary:regions_to_pyramids} is a reasonably good centerpoint for $X$. This is due to \Cref{obs:pyramid_decomposition}, and we formally prove it as follows.

\begin{observation}
\label{obs:m/2 centerpoint}
    Let $X \subseteq [0,1]^d$ with volume $\vol(X) \geq 0$, and consider a non-empty set $S \subseteq [d]$, an arbitrary $c\in \R^d$, and $m \geq 0$ with $\sum_{i\in S} \vol(P^+_i(c) \cap X) \geq m \vol(X)$ and $\vol (P^-_i(c) \cap X) \geq m \vol(X)$ for all $i\in S$. Then $c$ is an $m$-centerpoint of $X$.
\end{observation}
\begin{proof}
    Consider arbitrary $v \in \{-1,0,1\}^d \setminus \{0\}$ and its corresponding halfspace $H_v(c)$ around $c$. If we have $v_i \leq 0$ for some $i \in S$, then $\vol(H_v(c) \cap X) \geq \vol(P^-_i(c) \cap X)\geq m \vol(X)$ by \Cref{obs:pyramid_decomposition}. Otherwise, there is no such index, and thus $v_i=1$ for all $i\in S$, which implies $\vol(H_v(c) \cap X)\geq \sum_{i\in S}\vol (P^+_i(c) \cap X) \geq m \vol(X)$. Therefore, every halfspace around $c$ contains at least $m \vol(X)$ volume, concluding the proof.
\end{proof}

\subsection{Proof of \Cref{thm:approximate_centerpoint_theorem}}
\label{subsec:centerpoint_details}

With the preparations from \Cref{ssec:pulling}, we now move on to the proof of \Cref{thm:approximate_centerpoint_theorem}, which we will develop in several steps, according to the overview in \Cref{subsec:centerpoint_intuition}. We start with \Cref{lemma:balancing_pyramids} below, which states that we can approximately balance all the negative pyramids around a candidate $c \in \R^d$. We heavily rely on \Cref{lemma:pulling_operation} as a subroutine. Note that the runtime stated in \Cref{lemma:balancing_pyramids} could probably be improved by being a bit more careful. We opted for simplicity here, since improving the polynomial dependencies will not influence either of our final results (\Cref{thm:main_result_1} or \Cref{thm:main_result_2}) significantly.

\begin{figure}[ht]
    \centering
    \includegraphics[width=0.5\linewidth]{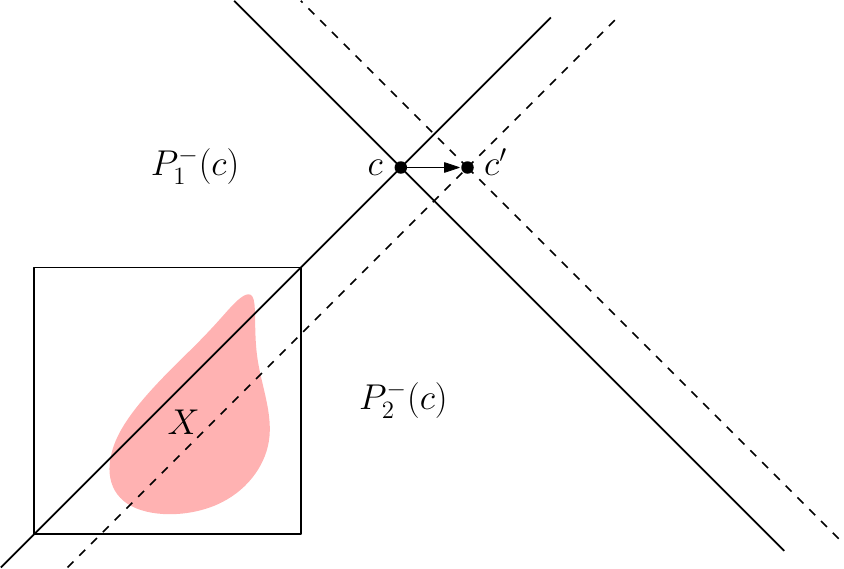}
    \caption{A two-dimensional sketch of one iteration in the proof of \Cref{lemma:balancing_pyramids}. Observe that $\vol(X \cap P^-_2(c))$ is much larger than $\vol(X \cap P^-_1(c))$, and thus we want to pull $c$ in the direction $(1, 0)$ (indicated by the arrow) to balance this out. It is not hard to see that in two dimensions, one iteration of the algorithm in \Cref{lemma:balancing_pyramids} actually suffices to balance the two negative pyramids.}
    \label{fig:balancing_neg}
\end{figure}

\begin{lemma}
\label{lemma:balancing_pyramids}
    Let $X \subseteq [0,1]^d $ with $\vol(X) \geq 0$ and $\epsilon > 0$ be arbitrary. Then we can find a point $c \in \R^d$ that satisfies $1 \leq c_i \leq 2^{\bigO(d^3 \log (1 / \epsilon))}$ and $\vol(P^-_i(c) \cap X) \geq \vol(X)/d-\epsilon$ for all $i \in [d]$ in time $\bigO(Td^7\log^2 1 / \epsilon)$, where we assume that we have access to an oracle that computes $\vol(X \cap P)$ for any pyramid $P$ in time $T$.
\end{lemma}

\begin{proof}
    We start by initializing $c = (1, \dots, 1) \in \R^d$, define $\pi_i = \vol(P^-_i(c) \cap X)$ for $i \in [d]$, and consider the potential function\footnote{We thank Hanwen Zhang for pointing us to this choice of potential function, which improved our analysis.} 
    $\Pi=\sum_{i\in [d]}(\pi_i-\frac{\vol(X)}{d})^2$. Throughout the algorithm, we will maintain the invariant $1 \leq c_i$ for all $i \in [d]$, and thus we will always have $\vol(P^+_i(c) \cap X) = 0$ for all $i \in [d]$ (by the definition of pyramids). Equivalently, we will always maintain $\sum_{i \in d} \pi_i = \vol(X)$.
    
    Observe that $\Pi$ is minimized for $\pi_i = \frac{\vol(X)}{d}$ for all $i \in [d]$. Thus, intuitively, decreasing the potential will balance the negative pyramids. Concretely, we have that $\Pi \leq \epsilon^2$ would imply $(\pi_i-\frac{\vol(X)}{d})^2\leq \epsilon^2$ and thus $-\epsilon\leq\pi_i-\frac{\vol(X)}{d}\leq \epsilon$ for all $i \in [d]$. This in turn would guarantee $\pi_i = \vol(P^-_i(c) \cap X)\geq\frac{\vol(X)}{d}-\epsilon$, which is our goal. 
    
    We will now iteratively improve the potential $\Pi$ by moving $c$ until we arrive at $\Pi \leq \epsilon^2$. Let us describe one improvement step: For this, we start by sorting the values of $\pi_1, \dots, \pi_d$ and assume from now on without loss of generality $\pi_1 \geq \pi_2 \geq \dots \geq \pi_d$. Let $k \in [d - 1]$ be the index that maximizes $\max_{k \in [d - 1]} (\pi_k-\pi_{k+1}) := \gamma$. We call $\gamma$ the largest gap.
    
    We define our pulling direction $u \in \{0, 1\}^d$ by setting $u_i = 0$ for $i \leq k$ and $u_i = 1$ for $i > k$. Clearly, pulling $c$ in this direction cannot decrease its coordinates, and thus our invariant $1 \leq c_i$ for all $i \in [d]$ is preserved. More concretely, we invoke \Cref{lemma:pulling_operation} to find a point $c' = c + \alpha u$ with 
    \[
        \frac{\gamma}{2} \geq \sum_{i > k} \pi'_i - \sum_{i > k} \pi_i \geq \frac{\gamma}{2} - \frac{\epsilon^2}{8 d^3}, 
    \]
    where we use $\pi'_i = \vol(P^-_i(c') \cap X)$ for the volume (w.r.t.\ $X$) of the $i$-th negative pyramid around $c'$. \Cref{lemma:pulling_operation} tells us that $c'$ can be found in time $\bigO(dT\log \frac{\|c\|d}{\epsilon})$. 
    
    By \Cref{lemma:pulling}, we get the guarantee $\delta_i := \pi_i - \pi'_i \leq 0$ for $i > k$ and $\delta_i := \pi_i - \pi'_i \geq 0$ for $i \leq k$: Indeed, the former follows from the first statement in \Cref{lemma:pulling} and the latter from the third statement in \Cref{lemma:pulling} (which implies that any point that a negative pyramid $P^-_i(c)$ with $i \leq k$ could potentially gain during the move would have to come from a positive pyramid, but these have zero volume w.r.t.\ X).
    Naturally, we also still have $\sum_{i \in [d]} \pi'_i = \vol(X)$. Thus, the move from $c$ to $c'$ effectively made the the first $k$ negative pyramids lighter and the remaining negative pyramids heavier. 

    Let $\Pi'$ be the potential of $c'$. We claim that our move decreased the potential, and in fact we will argue $\Pi - \Pi' \geq \frac{\gamma^2}{2} - \frac{\epsilon^2}{4d^3}$. We start by rewriting
    \begin{align*}
        \Pi-\Pi'&=\sum_{i\in[d]} \left( (\pi_i-\frac{\vol(X)}{d})^2-(\pi_i'-\frac{\vol(X)}{d})^2 \right)\\
        &=\sum_{i\in[d]} \left((\pi_i-\frac{\vol(X)}{d})^2-(\pi_i-\delta_i-\frac{\vol(X)}{d})^2 \right) \\
        &=\sum_{i\in[d]}  2\delta_i \left(\pi_i-\frac{\vol(X)}{d}\right)- \sum_{i\in[d]} \delta_i^2.
    \end{align*}
    Now observe that we have 
    \[
        \sum_{i\in[d]}  2\delta_i\left(\pi_i-\frac{\vol(X)}{d} \right) \geq 2 \left(\pi_k - \frac{\vol(X)}{d}\right) \left( \sum_{i\in[k]}  \delta_i \right) + 2 \left(\pi_{k + 1} - \frac{\vol(X)}{d} \right) \left( \sum_{i = k + 1}^d  \delta_i  \right)
    \]
    since $\delta_i \geq 0$ for $i \in [k]$ and $\delta_i \leq 0$ for $i > k$. By $\sum_{i \in [d]}\pi_i = \vol(X) = \sum_{i \in [d]}\pi'_i$, we have
    \[
        \sum_{i\in[k]}  \delta_i = - \sum_{i = k + 1}^d  \delta_i
    \]
    and thus get
    \begin{align*}
        \sum_{i\in[d]}  2\delta_i\left(\pi_i-\frac{\vol(X)}{d}\right) &\geq 2 \left(\pi_k - \frac{\vol(X)}{d}\right) \left( \sum_{i\in[k]} \delta_i \right)  - 2 \left(\pi_{k + 1} - \frac{\vol(X)}{d}\right) \left( \sum_{i\in[k]}  \delta_i  \right) \\
        &= 2 (\pi_k - \pi_{k + 1}) \sum_{i\in[k]}  \delta_i  \\
        &\geq 2  \gamma \sum_{i\in[k]}  \delta_i.
    \end{align*}
    Next, observe also 
    \[
        \sum_{i\in[d]} \delta_i^2 \leq \left(\sum_{i\in[k]} |\delta_i| \right)^2  + \left(\sum_{i = k + 1}^d |\delta_i| \right)^2 \leq 2\frac{\gamma^2}{2^2} = \frac{\gamma^2}{2}
    \]
    due to $\sum_{i > k} |\delta_i| =  \sum_{i \in [k]} |\delta_i| \leq \frac{\gamma}{2}$. We use this to conclude 
    \[
        \Pi - \Pi' \geq 2  \gamma \sum_{i\in[k]}  \delta_i - \frac{\gamma^2}{2} \geq 2 \gamma (\frac{\gamma}{2} - \frac{\epsilon^2}{8d^3}) - \frac{\gamma^2}{2} = \frac{\gamma^2}{2} - \frac{\epsilon^2}{4d^3},
    \]
    where we used $\sum_{i \in [k]} \delta_i \geq \frac{\gamma}{2} - \frac{\epsilon^2}{8d^3}$ and $\gamma \leq \vol(X) \leq 1$.
    
    To summarize, this calculation shows that if the largest gap in $\pi_1, \dots, \pi_d$ is $\gamma$, one iteration of pulling $c$ decreases the potential by at least $\frac{\gamma^2}{2} - \frac{\epsilon^2}{4d^3}$. Since the potential of $\Pi$ with largest gap $\gamma$ is generally upper bounded by $ \Pi \leq d(d\gamma)^2=d^3\gamma^2$, we get a relative decrease of potential of 
    \[
        \Pi' \leq \Pi - \frac{\gamma^2}{2} + \frac{\epsilon^2}{4d^3} = (1 - \frac{\gamma^2}{2 \Pi}) \Pi + \frac{\epsilon^2}{4d^3} \leq (1 - \frac{1}{2d^3}) \Pi + \frac{\epsilon^2}{4d^3}
    \]
    with every iteration. Using an upper bound of $\vol(X)^2 \leq 1$ for the initial potential, we calculate that after $t$ iterations, we obtain an upper bound of 
    \[
        \Pi \leq (1 - \frac{1}{2d^3})^t \vol(X)^2 + \frac{\epsilon^2}{4d^3} \sum_{i = 0}^{t - 1} ( 1 - \frac{1}{2d^3})^i \leq (1 - \frac{1}{2d^3})^t + 2d^3 \frac{\epsilon^2}{4d^3} \leq (1 - \frac{1}{2d^3})^t + \frac{\epsilon^2}{2}
    \]
    on the potential. Thus, by choosing $t \geq 2d^3 \log \frac{2}{\epsilon^2}$, we get 
    \[
        \Pi \leq  \exp( - \frac{1}{2d^3} 2d^3 \log \frac{2}{\epsilon^2} ) + \frac{\epsilon^2}{2} = \epsilon^2,
    \]
    as desired. 

    We perform $\bigO(d^3 \log \frac{1}{\epsilon})$ iterations. Each iteration can be performed in $\bigO(dT\log \frac{\|c\|d}{\epsilon})$ time, and in particular increases the coordinates of $c$ by at most $\bigO(\|c\|)$ (due to the use of \Cref{lemma:pulling_upper_bound} in \Cref{lemma:pulling_operation}). Thus, we get an overall bound of $\|c\| \leq 2^{\bigO(d^3 \log \frac{1}{\epsilon})}$, which means that the overall runtime is bounded by $\bigO(Td^7 \log^2 1 / \epsilon)$.
\end{proof}

Before we can conclude \Cref{thm:approximate_centerpoint_theorem}, it remains to explain how we can project to the box $[0, 1]^d$, which will be done in \Cref{lemma:centerpoint_projection_on_cube} below.

\begin{figure}[ht]
    \centering
    \includegraphics[width=0.3\linewidth]{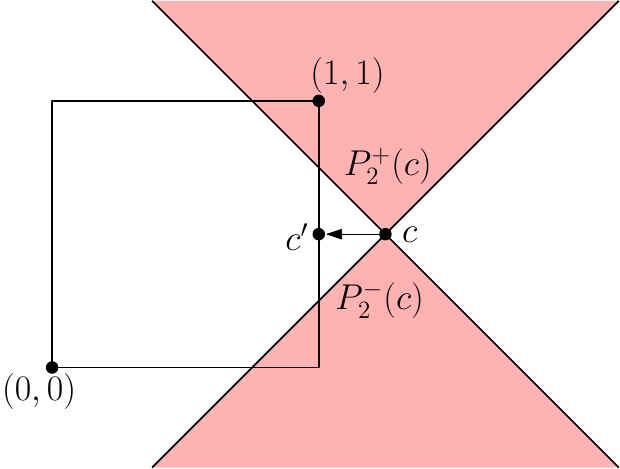}
    \caption{A two-dimensional example of the projection in the proof of \Cref{lemma:centerpoint_projection_on_cube}.}
    \label{fig:projection}
\end{figure}

\begin{lemma}
\label{lemma:centerpoint_projection_on_cube}
    Let $X \subseteq [0, 1]^d$ with $\vol(X) \geq 0$ be arbitrary. Assume that we are given a point $c \in \R^d$ with $\sum_{i\in[d]} \vol(P^+_i(c) \cap X)\geq m$ and $\vol(P^-_i(c) \cap X)\geq m$ for all $i \in [d]$ and some $m \geq 0$. Then we can find a point $c' \in [0, 1]^d$ with $\sum_{i\in S} \vol(P^+_i(c') \cap X)\geq m$ and $\vol(P^-_i(c') \cap X) \geq m $ for all $i$ in some non-empty subset $S \subseteq [d]$ in time $O(d)$.
\end{lemma}
\begin{proof}
    Observe first that the case $m = 0$ is trivial, and thus assume from now on $m > 0$. Consider the projection $c' \in [0,1]^d$ of $c$ to the cube $[0, 1]^d$, defined as  
    \[
        c'_i = \begin{cases}
            0 & \text{ if } c_i < 0 \\
            c_i & \text{ if } 0 \leq c_i \leq 1 \\
            1 & \text{ if } 1 < c_i
        \end{cases}
    \]
    for all $i \in [d]$, and let $S = \{i \in [d] \mid c_i = c'_i\}$ be the subset of indices for which no projection was necessary. We claim that $S$ is non-empty and that $\sum_{i \in S} \vol(P^+_i(c') \cap X)\geq m$ and $\vol(P^-_i(c') \cap X)\geq m$ for all $i \in S$.

    Since $m > 0$, there is at least one positive pyramid around $c$ with non-zero volume, and all negative pyramids around $c$ have non-zero volume as well. Therefore, there is at least one coordinate, say $i \in [d]$, such that both $\vol(P^+_i(c) \cap X ) > 0$ and $\vol(P^-_i(c) \cap X ) > 0$. But this implies $0 \leq c_i \leq 1$ by the definition of pyramids and the assumption $X \subseteq [0, 1]^d$. Therefore, the set $S$ is nonempty. More generally, this observation shows that
    \[
        \sum_{i \in S} \vol(P^+_i(c) \cap X) = \sum_{i \in [d]} \vol(P^+_i(c) \cap X) \geq m,
    \]
    since $\vol(P^-_j(c) \cap X) \geq m > 0$ actually implies $c_j \geq 0$ for all $j \in [d]$.
    
    Next, consider arbitrary $i \notin S$. Since $i \notin S$ and since we must have $c_i \geq 0$, we actually conclude $c_i > 1$ and thus will have projected $c_i$ to $c'_i = 1$. By the definition of pyramids, we conclude that the positive pyramid $P^+_i(c')$ around $c'$ must satisfy $\vol(P^+_i(c') \cap X ) = 0$. 

    Next, observe that projecting one coordinate at a time corresponds to performing at most $d$ iterative pullings on $c$: Once in the negative direction of $i$ for all $i \notin S$. By \Cref{lemma:pulling}, we know that while pulling $c$ in the direction $u$ with $u_i = -1$ and $u_j = 0$ for $j \neq i$, $\vol(P^-_i(c) \cap X)$ may decrease, $\vol(P^+_i(c) \cap X)$ may increase, and all remaining pyramids may gain volume from $P^-_i(c)$ and lose volume to $P^+_i(c)$. However, due to our observation $\vol(P^+_i(c') \cap X) = 0$, we know that the latter never happens. Thus, none of the pyramids $P^+_j(c)$ or $P^-_j(c)$ for $j \in S$ can lose volume during any of the pullings and hence the whole projection. This concludes the proof.
\end{proof}

Finally, we are ready to conclude \Cref{thm:approximate_centerpoint_theorem}. 

\thmapproxcenterpoint*
\begin{proof}
    Fix $\epsilon = \frac{\vol(X)}{6d}$. By \Cref{lemma:balancing_pyramids}, we can find $c \in \R^d$ with $\vol(P_i^-(c) \cap X) \geq \frac{\vol(X)}{d} - \epsilon$ for all $i \in [d]$ in time $\bigO(d^7T\log^2 1 / \epsilon)$. 
    We then use \Cref{corollary:regions_to_pyramids} to compute $c' \in \R^d$ with 
    \[
        \sum_{i \in [d]} \vol(P^+_i(c') \cap X)\geq \frac{\vol(X)}{2d} - \frac{3}{2}\epsilon \quad \text{ and } \quad \vol(P^-_i(c') \cap X) \geq \frac{\vol(X)}{2d} - \frac{\epsilon}{2}
    \]
    for all $i \in [d]$. This uses again at most $\bigO(dT \log \frac{\|c\|d}{\epsilon})$ time, where $\|c\| \leq 2^{\bigO(d^3 \log (1 / \epsilon))}$ is guaranteed by \Cref{lemma:balancing_pyramids}.
    Finally, we use \Cref{lemma:centerpoint_projection_on_cube} to project $c'$ to the cube $[0, 1]^d$, yielding $c'' \in [0, 1]^d$ with 
    \[
        \sum_{i\in S} \vol(P^+_i(c'') \cap X)\geq \frac{\vol(X)}{2d} - \frac{3}{2}\epsilon \quad \text{ and } \quad \vol(P^-_i(c'') \cap X) \geq \frac{\vol(X)}{2d} - \frac{\epsilon}{2}
    \]
    for all $i$ in some non-empty subset $S \subseteq [d]$. By \Cref{obs:m/2 centerpoint} and our choice of $\epsilon = \frac{\vol(X)}{6d}$, we get that $c''$ is a $\frac{1}{4d}$-centerpoint of $X$. The last step runs in time $\bigO(d)$ and thus the overall runtime is $\bigO(d^7T\log^2 1 / \epsilon)$ which is $\bigO(d^7T\log^2 \frac{d}{\vol X})$.
\end{proof}

Note that the centerpoint quality in \Cref{thm:approximate_centerpoint_theorem} could be improved: In particular, with the current proof we can get arbitrarily close to $\frac{1}{2d}$ by choosing $\epsilon$ sufficiently small.

\section{Volume Computation}
\label{sec:volume_computation}

The main goal of this section is to prove \Cref{lemma:volume_computation}. We start by recalling some important tools form computational geometry, before we combine them with a simple observation on pyramids to get the desired result.

\paragraph*{Volume of Convex Polyhedra} As a crucial ingredient, we will use the fact that the volume of a given bounded convex polyhedron described by a linear system $A x  \leq b$ in $\R^d$ can be computed exactly in time $N^{\bigO(d)}$, where $N$ is the encoding size of the linear system. Such exact algorithms were studied by a number of authors~\cite{ barvinokComputingVolumeCounting1993, cohenTwoAlgorithmsDetermining1979, lasserreAnalyticalExpressionAlgorithm1983, lasserreLaplaceTransformAlgorithm2001, lawrencePolytopeVolumeComputation1991}. For example, the algorithm by Cohen and Hickey~\cite{cohenTwoAlgorithmsDetermining1979} decomposes the polyhedron into simplices (potentially exponentially many) and adds up the volumes of all those simplices to get the result. 

\paragraph*{Hyperplane Arrangements} We briefly recapitulate some facts on hyperplane arrangements that will be important for us\footnote{For more background on the topic, see~\cite{edelsbrunnerAlgorithmsCombinatorialGeometry1987}.}: Consider $n$ hyperplanes $h_1, \dots, h_n$ in $\R^d$. Together, they induce a hyperplane arrangement, and the $d$-dimensional cells of this arrangement are convex polyhedra. In particular, $h_1, \dots, h_n$ decompose the cube $[0, 1]^d$ into at most $n^{\bigO(d)}$ bounded convex polyhedra\footnote{We want to thank Hanwen Zhang for pointing us to this bound.}.
Moreover, there is a deterministic algorithm that computes the hyperplane arrangement induced by $h_1, \dots, h_n$ in time $n^{\bigO(d)}$~(see \cite{edelsbrunnerConstructingArrangements1987}). This implies that in $n^{\bigO(d)}$ time, it is possible to compute a linear system of inequalities describing each of the $n^{\bigO(d)}$ bounded convex polyhedra that together make up $[0, 1]^d$ (note that each linear system consists of  at most $n + 2d$ inequalities, since each hyperplane and boundary of the cube $[0, 1]^d$ can add at most one inequality).

\paragraph*{Bounding Hyperplanes of Pyramids} To see why hyperplane arrangements are relevant for us, consider the following observation on pyramids: Let $P^s_i(x)$ be an arbitrary pyramid around $x$ in $\R^d$, and recall that it is defined by $P^s_i(x) = \{ y \in \R^d \mid s(y_i - x_i) = \|y - x\| \}$. In particular, $P^s_i(x)$ can be written as 
\[
    P^s_i(x) = \left( \bigcap_{j \neq i} \{ y \in \R^d \mid s(y_i - x_i) \geq y_j - x_j  \} \right) \cap \left( \bigcap_{j \neq i} \{ y \in \R^d \mid s(y_i - x_i) \geq -y_j + x_j \} \right)
\]
and is thus bounded by exactly $2(d - 1)$ hyperplanes: Indeed, each of those hyperplanes has the form $\pm y_i \pm y_j \geq r_x$ for $i, j \in [d]$ with $i \neq j$ and some value $r_x \in \R$ that depends solely on $x$. With this in place, we can now prove \Cref{lemma:volume_computation}.  

\begin{figure}[ht]
    \centering
    \includegraphics[width=0.4\linewidth]{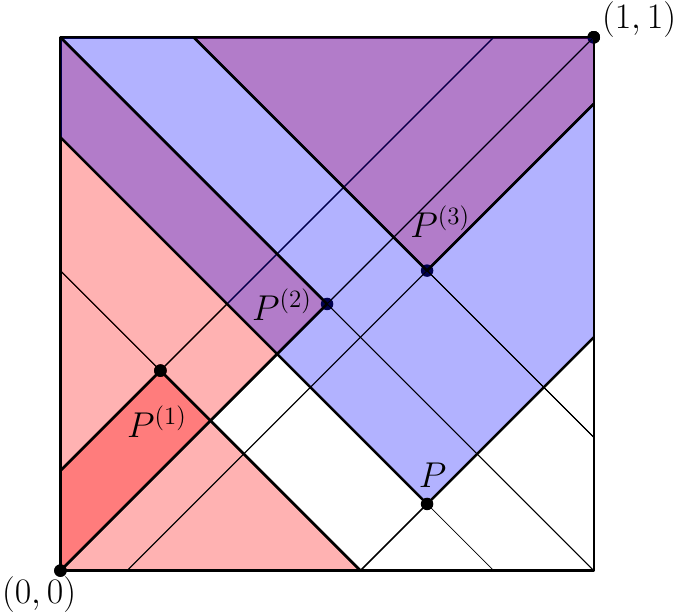}
    \caption{Example of the proof of \Cref{lemma:volume_computation} in the plane. Each of the four pyramids $P^{(1)}, P^{(2)}, P^{(3)}$ (in red) and $P$ (in blue) is bounded by $2$ hyperplanes. The induced hyperplane arrangement decomposes the cube $[0, 1]^d$ into $24$ bounded convex polyhedra. We can compute $\vol(([0, 1]^2 \setminus (P^{(1)} \cup P^{(2)} \cup P^{(3)})) \cap P)$ by summing up the volumes of all polyhedra that are contained in $P$ but not in any of $P^{(1)}, P^{(2)}, P^{(3)}$.}
    \label{fig:decomposition}
\end{figure}

\lemmavolumecomputation*

\begin{proof}
    Consider an arbitrary valid search space $X = [0, 1]^d \setminus \bigcup_{i \in [n]} P^{(i)}$, and let $P$ be our query pyramid, i.e.\ we need to find the volume of the set
    \[
        X \cap P = \left([0, 1]^d \setminus \bigcup_{i \in [n]} P^{(i)} \right) \cap P.
    \]
    By our observation above on pyramids, each of the pyramids is bounded by $2(d - 1)$ hyperplanes. Thus, we get a total of $2(n + 1)(d - 1)$ hyperplanes from the $n + 1$ pyramids, which together decompose the cube $[0, 1]^d$ into at most $(nd)^{\bigO(d)}$ bounded convex polyhedra. Since we have considered all the bounding hyperplanes of all the pyramids, the interior of each bounded convex polyhedron in this decomposition of $[0, 1]^d$ is either fully contained or fully excluded in any of $n + 1$ pyramids $P^{(1)}, \dots, P^{(n)}, P$. Thus, we can compute 
    \[
        \vol \left( \left( [0, 1]^d \setminus \bigcup_{i \in [n]} P^{(i)} \right) \cap P \right)
    \]
    by summing up the volume of each convex polyhedron in the decomposition that is contained in $P$ but not in any of the $n$ pyramids $P^{(1)}, \dots, P^{(n)}$. From our preparation above, we know that we can iterate over all of the polyhedra in time $(nd)^{\bigO(d)}$ and compute the volume of each in time $(nd)^{\bigO(d)}$, yielding an overall runtime of at most $(nd)^{\bigO(d)}$, as desired.
\end{proof}

\paragraph*{Bit Complexity}

As noted in \Cref{ssec:proof_thm_1}, the bit complexity of the computed volume could in theory be exponential in $d$, i.e.\ assuming that all of the pyramids are described by points of bit complexity $N$, the result could have complexity $N^{\bigO(d)}$. This is already inherent in the algorithm of Cohen and Hickey~\cite{cohenTwoAlgorithmsDetermining1979}, where it comes from summing up the rational volumes of exponentially many simplices (in particular, the denominator of the result could require many bits). We then further sum these volumes up over exponentially many polyhedra induced by our hyperplane arrangement. Overall, this still yields ``only'' an exponential bit complexity for the result, and thus our stated runtime bound accommodates for this. We repeat that the volume computed by our algorithms via \Cref{lemma:volume_computation} is only ever used for comparison in \Cref{alg:algo_1}, and thus the blow-up in terms of bit complexity is not further propagated through \Cref{alg:algo_1}. Note that one might also be able to avoid the blow-up altogether by appropriately rounding the results (in such a way that the correct results for comparisons of volume in \Cref{sec:finding_centerpoint} are still guaranteed).

\section{Decomposition for $\ell_\infty$-Contraction and \Cref{thm:main_result_2}}
\label{sec:decomposition}

In this section, we explain the remaining details that are needed to prove \Cref{thm:main_result_2}. For this, recall the problems $\contr(d, \epsilon, \lambda)$ and $\nonexp^\dagger(d, \epsilon)$ that we already introduced in \Cref{ssec:proof_thm_2}, and also consider the intermediate problem $\nonexp(d, \epsilon)$:
\begin{itemize}
    \item $\contr(d, \epsilon, \lambda)$ is the problem of finding an $\epsilon$-approximate fixed point of a $\lambda$-contracting function $f : [0, 1]^d \rightarrow [0, 1]^d$.
    \item $\nonexp(d, \epsilon)$ is the problem of finding an $\epsilon$-approximate fixed point of a non-expansive function $f : [0, 1]^d \rightarrow [0, 1]^d$.
    \item $\nonexp^\dagger(d, \epsilon)$ is the problem of finding an $\epsilon$-approximate fixed point of a non-expansive function $f : [0, 1]^d \rightarrow [-\epsilon, 1 + \epsilon]^d$.
\end{itemize}

We also restate the decomposition theorem from~\cite{chenQuadraticSpeedupComputing2026}.

\thmdecompositiontheorem*

\subsection{Projection and Choice of Domain}

Technically, a more general version of $\nonexp^\dagger(d, \epsilon)$ is used in~\cite{chenQuadraticSpeedupComputing2026}: Indeed, they allow non-expansive functions $f : [a_1, b_1] \times \dots \times [a_d, b_d] \rightarrow [a_1 - \epsilon, b_1 + \epsilon] \times \dots \times [a_d - \epsilon, b_d + \epsilon]$ for arbitrary $a, b \in [0, 1]^d$ with $a \leq b$. This is important because their algorithm for $\nonexp^\dagger(d_1 + d_2, \epsilon)$ actually constructs specific domains on which it then calls the subroutine for $\nonexp^\dagger(d_1, \epsilon)$. However, any instance on a domain $[a_1, b_1] \times \dots \times [a_d, b_d] \subseteq [0, 1]^d$ can easily be reduced to an instance on the cube $[0, 1]^d$ (see below), and thus we can safely restrict ourselves to the case of the cube $[0, 1]^d$ (which is what we did so far).

For the sake of completeness, let us present the reduction that allows us to assume domain $[0, 1]^d$ without loss of generality. Given two points $a, b \in \R^d$ with $a \leq b$, let $\llbracket a, b \rrbracket \subseteq \R^d$ denote the box $[a_1, b_1] \times \dots \times [a_d, b_d]$. The following tool (which we refer to as projection to the box $\llbracket a, b \rrbracket$) will also be useful later for us (and in fact, we have already used an ad-hoc version of this in the proof of \Cref{lemma:centerpoint_projection_on_cube}).

\begin{definition}[Projection]
\label{def:projection}
    Let $x \in \R^d$ and $\llbracket a, b \rrbracket \subseteq \R^d$ be arbitrary. We define the projection $\proj_{\llbracket a, b \rrbracket}(x) \in \llbracket a, b \rrbracket$ of $x$ to $\llbracket a, b \rrbracket$ as
    \[
        \proj_{\llbracket a, b \rrbracket}(x)_i = \begin{cases}
            a_i & \text{ if } x_i < a_i \\
            x_i & \text{ if } a_i \leq x_i \leq b_i \\
            b_i & \text{ if } b_i < x_i
        \end{cases}
    \]
    for all $i \in [d]$.
\end{definition}

\begin{observation}
\label{obs:generality_by_projection}
    Let $\llbracket a, b \rrbracket \subseteq \R^d$ be arbitrary. We have $\|\proj_{\llbracket a, b \rrbracket}(x) - \proj_{\llbracket a, b \rrbracket}(y)\| \leq \|x - y\|$
     for all $x, y \in \R^d$.
\end{observation}
\begin{proof}
    This can be easily verified by doing the case distinction for each coordinate separately. 
\end{proof}

With \Cref{obs:generality_by_projection}, it is now easy to show that we can always assume domain $[0, 1]^d$ for our problems: Indeed, consider a non-expansive function $f : \llbracket a, b \rrbracket \rightarrow \llbracket a - \epsilon \One, b + \epsilon \One\rrbracket$. Define a new function $g : [0, 1]^d \rightarrow [-\epsilon, 1 + \epsilon]^d$ as $g(x) := f(\proj_{\llbracket a, b \rrbracket}(x))$. Then \Cref{obs:generality_by_projection} tells us that $g$ is non-expansive as well, since
\[
    \|g(x) - g(y)\| = \|f(\proj_{\llbracket a, b \rrbracket}(x)) - f(\proj_{\llbracket a, b \rrbracket}(y))\| \leq \|\proj_{\llbracket a, b \rrbracket}(x) - \proj_{\llbracket a, b \rrbracket}(y)\| \leq \|x - y\| 
\]
for all $x, y \in [0, 1]^d$. Moreover, let $x^\star$ be an $\epsilon$-approximate fixed point of $g$, and consider arbitrary $i \in [d]$. If we have $x^\star_i \in [a_i, b_i]$, we get 
\[
    |f(\proj_{\llbracket a, b \rrbracket}(x^\star))_i - \proj_{\llbracket a, b \rrbracket}(x^\star)_i| = |g(x^\star)_i - x^\star_i| \leq \epsilon.
\]
If instead $x^\star_i < a_i$, then by $|g(x^\star)_i - x^\star_i| \leq \epsilon$ and $g(x^\star)_i = f(\proj_{\llbracket a, b \rrbracket}(x^\star))_i \in [a_i - \epsilon, b_i + \epsilon]$, we can deduce $|f(\proj_{\llbracket a, b \rrbracket}(x^\star))_i - \proj_{\llbracket a, b \rrbracket}(x^\star)_i| = |g(x^\star)_i - a_i| \leq \epsilon$. Analogously, we also get $|f(\proj_{\llbracket a, b \rrbracket}(x^\star))_i - \proj_{\llbracket a, b \rrbracket}(x^\star)_i| \leq \epsilon$ if we assume $x^\star_i > b_i$. We conclude that $\proj_{\llbracket a, b \rrbracket}(x^\star)$ is an $\epsilon$-approximate fixed point of $f$.

Of course, this reduction could also be applied analogously to the problems $\contr(d, \epsilon, \lambda)$ and $\nonexp(d, \epsilon)$.

\subsection{From Contraction to Non-Expansion and Back}
\label{ssec:reductions}

Clearly, $\nonexp^\dagger(d, \epsilon)$ is more general than $\nonexp(d, \epsilon)$, and $\nonexp(d, \epsilon)$ is in turn more general than $\contr(d, \epsilon, \lambda)$. It turns out that reductions in the reverse direction exist as well. 

For example, a simple and efficient reduction from $\nonexp(d, \epsilon)$ to $\contr(d, \epsilon / 2, 1 - \epsilon / 2)$ was observed in~\cite{chenComputingFixedPoint2025}: Indeed, given a non-expansive function $f : [0, 1]^d \rightarrow [0, 1]^d$, consider a new function $g : [0, 1]^d \rightarrow [0, 1]^d$
defined by $g(x) = (1 - \frac{\epsilon}{2})f(x)$ for all $x \in [0, 1]^d$. It is then easy to check that $g$ is $(1 - \frac{\epsilon}{2})$-contracting and that any $\frac{\epsilon}{2}$-approximate fixed point of $g$ is also an $\epsilon$-approximate fixed point of $f$.

Moreover, in the case $d = 2$, a reduction from $\nonexp^\dagger(2, \epsilon)$ to $\nonexp(2, \epsilon)$ was given in~\cite{chenQuadraticSpeedupComputing2026}. In \Cref{lemma:reduction_to_nonexpansion} below, we extend this reduction to higher dimensions, albeit at the cost of decreasing $\epsilon$ by a factor of $2$ again.

\begin{figure}[ht]
    \centering
    \includegraphics[width=0.4\linewidth]{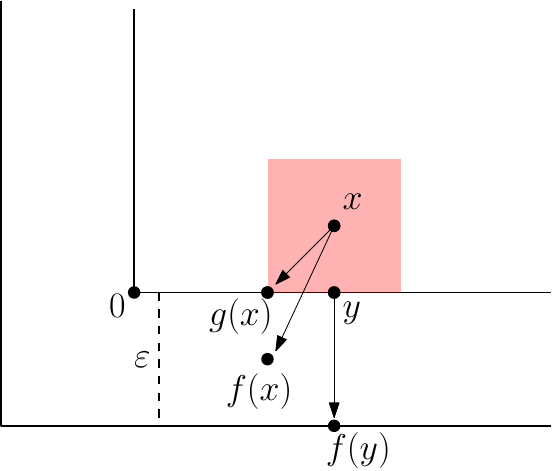}
    \caption{A sketch of the reduction in \Cref{lemma:reduction_to_nonexpansion} showing an example of the points $x, y, f(x), g(x), f(y)$ from the proof. The $\ell_\infty$-ball $B^\infty_{\frac{\epsilon}{2}}(x)$ of radius $\frac{\epsilon}{2}$ around $x$ is shown in red.}
    \label{fig:reduction}
\end{figure}

\begin{lemma}
\label{lemma:reduction_to_nonexpansion}
    If $\nonexp(d, \epsilon / 2)$ can be solved in $q$ queries and time $t$, then $\nonexp^\dagger(d, \epsilon)$ can be solved in $q + 1$ queries and time $\bigO(t)$.
\end{lemma}

\begin{proof}
    Let $f : [0, 1]^d \rightarrow [-\epsilon, 1 + \epsilon]^d$ be a non-expansive function. Consider the truncated function $g : [0, 1]^d \rightarrow [0, 1]^d$ defined as 
    \[
        g(x) = \proj_{[0, 1]^d}(f(x))
    \]
    for all $i \in [d]$. By \Cref{obs:generality_by_projection}, $g$ is still non-expansive and thus we can use an algorithm for $\nonexp(d, \epsilon / 2)$ to find an $\epsilon/2$-approximate fixed point $x \in [0, 1]^d$ (i.e.\ $\|x - g(x)\| \leq \epsilon / 2$). Given $x$, we then compute the point $y \in  [0, 1]^d$ defined by 
    \[
        y_i = \begin{cases}
            x_i &\text{ if } f_i(x) \in [0, 1] \\
            0 &\text{ if } f_i(x) < 0 \\
            1 &\text{ if } f_i(x) > 1 \\
        \end{cases}
    \]
    for all $i \in [d]$. We claim that $y$ is an $\epsilon$-approximate fixed point of $f$, i.e.\ we need to prove $\|f(y) - y\| \leq \epsilon$. For this, observe first that since $x$ is $\epsilon / 2$-approximate, we have $\|y- x\| \leq \epsilon / 2$. Now, for arbitrary $i \in [d]$ with $f_i(x) \in [0, 1]$, we indeed have 
    \begin{align*}
        |f_i(y) - y_i| &\leq |f_i(y) - f_i(x)| + |f_i(x) - y_i|  \\
        &= |f_i(y) - f_i(x)| + |g_i(x) - x_i| \\
        &\leq \|f(y) - f(x)\| + \epsilon / 2 \\
        &\leq \|y - x\| + \epsilon / 2 \leq \epsilon.
    \end{align*}
    Moreover, for any $i \in [d]$ with $f_i(x) < 0$, we observe that we directly get $|f_i(y) - y_i| = |f_i(y) - 0| \leq \epsilon$ if $f_i(y) < 0$. If, on the other hand, we have $f_i(y) \geq 0$, then using $\|y - x\| \leq \epsilon / 2$ and non-expansion we get $|f_i(y) - y_i| = |f_i(y) - 0| \leq |f_i(y) - f_i(x)| \leq \epsilon / 2 \leq \epsilon$ since $f_i(x) < 0$. The case of $i \in [d]$ with $f_i(x) > 1$ is symmetric and we conclude $\|f(y) - y\| \leq \epsilon$.
\end{proof}

\subsection{Combining Decomposition and Reductions}

By combining the decomposition theorem (\Cref{thm:decomposition_theorem}) with \Cref{lemma:reduction_to_nonexpansion} and the reduction from $\nonexp(d, \epsilon)$ to $\contr(d, \epsilon / 2, 1 - \epsilon / 2)$, we are now able to conclude the following result. 

\squarerootlemma*

\begin{proof}
    First recall the reduction from $\nonexp(d, 2\epsilon)$ to $\contr(d, \epsilon, 1 - \epsilon)$ described in \Cref{ssec:reductions}, which tells us that we can also solve $\nonexp(d, 2\epsilon)$ in $q$ queries and $\bigO(t)$ time. Moreover, by \Cref{lemma:reduction_to_nonexpansion}, this means that we can solve $\nonexp^\dagger(d, 4\epsilon)$ in $q + 1$ queries and $\bigO(t)$ time.
    
    Now consider an arbitrary instance $f : [0, 1]^{d^2} \rightarrow [-\epsilon, 1 + \epsilon]^{d^2}$ of $\nonexp^\dagger(d^2, 4\epsilon)$. We split the $d^2$ dimensions into $d$ blocks of dimension $d$ each. By using \Cref{thm:decomposition_theorem} recursively, we thus get an algorithm that solves $\nonexp^\dagger(d^2, 4\epsilon)$ in $(q + 1)^d$ queries and $\sum_{i = 0}^d (cq^2)^i t \leq q^{\bigO(d)} t$ time (where $c$ is a constant).
\end{proof}

In combination with our \Cref{thm:main_result_1}, we thus conclude \Cref{thm:main_result_2}.

\thmmainresulttwo*

\begin{proof}
    Since $\contr(d, \epsilon, \lambda)$ reduces to $\nonexp^\dagger(d, \epsilon)$, it suffices to solve the latter problem. Without loss of generality, assume that we can split the $d$ dimensions into $\sqrt{d}$ blocks of size $\sqrt{d}$ each (this can be achieved by adding additional dimensions until $d$ is square). By \Cref{lemma:square-root-decomposition}, this means that we get an algorithm using $(q + 1)^{\sqrt{d}}$ queries and $q^{\bigO(\sqrt{d})}t$ time, where $q$ and $t$ denote the query and time upper bounds of solving $\contr(\sqrt{d}, \frac{\epsilon}{4}, 1 - \frac{\epsilon}{4})$. By \Cref{thm:main_result_1}, we can solve $\contr(\sqrt{d}, \frac{\epsilon}{4}, 1 - \frac{\epsilon}{4})$ in $q \leq \poly(d, \log \frac{1}{\epsilon})$ queries and $t \leq (\log \frac{1}{\epsilon})^{\bigO(\sqrt{d}\log d)}$ time.
    Therefore, we get an overall runtime of $(\log \frac{1}{\epsilon})^{\bigO(\sqrt{d} \log d)}$.
\end{proof}

Of course, as is also the case with \Cref{thm:main_result_1}, \Cref{thm:main_result_2} actually applies to the more general problem $\nonexp^\dagger(d, \epsilon)$, and not just $\contr(d, \epsilon, \lambda)$ (our proof is even formulated in this way). 

\bibliographystyle{alphaurl}
\bibliography{references}

@misc{chenQuadraticSpeedupComputing2026,
    title = {Quadratic {Speedup} for {Computing} {Contraction} {Fixed} {Points}},
    doi = {10.48550/arXiv.2602.10296},
    urldate = {2026-02-28},
    publisher = {arXiv},
    author = {Chen, Xi and Li, Yuhao and Yannakakis, Mihalis},
    month = feb,
    year = {2026},
    note = {arXiv:2602.10296 [cs]
version: 1},
}

@article{chenComputingFixedPoint2025,
    title = {Computing a {Fixed} {Point} of {Contraction} {Maps} in {Polynomial} {Queries}},
    volume = {72},
    issn = {0004-5411},
    doi = {10.1145/3744738},
    number = {4},
    urldate = {2025-10-28},
    journal = {J. ACM},
    author = {Chen, Xi and Li, Yuhao and Yannakakis, Mihalis},
    month = jul,
    year = {2025},
    pages = {25:1--25:20},
}

@inproceedings{etessamiTarskiTheoremSupermodular2020,
    address = {Dagstuhl, Germany},
    series = {Leibniz {International} {Proceedings} in {Informatics} ({LIPIcs})},
    title = {Tarski’s {Theorem}, {Supermodular} {Games}, and the {Complexity} of {Equilibria}},
    volume = {151},
    isbn = {978-3-95977-134-4},
    issn = {1868-8969},
    doi = {10.4230/LIPIcs.ITCS.2020.18},
    urldate = {2023-08-26},
    booktitle = {11th {Innovations} in {Theoretical} {Computer} {Science} {Conference} ({ITCS} 2020)},
    publisher = {Schloss Dagstuhl–Leibniz-Zentrum fuer Informatik},
    author = {Etessami, Kousha and Papadimitriou, Christos and Rubinstein, Aviad and Yannakakis, Mihalis},
    editor = {Vidick, Thomas},
    year = {2020},
    pages = {18:1--18:19},
}

@article{ballCubeSlicingRn1986,
    title = {Cube {Slicing} in $\mathbb{R}^n$},
    volume = {97},
    issn = {0002-9939},
    doi = {10.2307/2046239},
    number = {3},
    urldate = {2026-03-16},
    journal = {Proceedings of the American Mathematical Society},
    publisher = {American Mathematical Society},
    author = {Ball, Keith},
    year = {1986},
    pages = {465--473},
}

@article{lasserreAnalyticalExpressionAlgorithm1983,
    title = {An analytical expression and an algorithm for the volume of a convex polyhedron {inRn}},
    volume = {39},
    issn = {1573-2878},
    doi = {10.1007/BF00934543},
    language = {en},
    number = {3},
    urldate = {2026-03-16},
    journal = {Journal of Optimization Theory and Applications},
    author = {Lasserre, J. B.},
    month = mar,
    year = {1983},
    pages = {363--377},
}

@article{lawrencePolytopeVolumeComputation1991,
    title = {Polytope volume computation},
    volume = {57},
    issn = {0025-5718, 1088-6842},
    doi = {10.1090/S0025-5718-1991-1079024-2},
    language = {en},
    number = {195},
    urldate = {2026-03-16},
    journal = {Mathematics of Computation},
    author = {Lawrence, Jim},
    year = {1991},
    pages = {259--271},
}

@article{lasserreLaplaceTransformAlgorithm2001,
    title = {A {Laplace} transform algorithm for the volume of a convex polytope},
    volume = {48},
    issn = {0004-5411},
    doi = {10.1145/504794.504796},
    number = {6},
    urldate = {2026-03-16},
    journal = {J. ACM},
    author = {Lasserre, Jean B. and Zeron, Eduardo S.},
    month = nov,
    year = {2001},
    pages = {1126--1140},
}

@article{barvinokComputingVolumeCounting1993,
    title = {Computing the volume, counting integral points, and exponential sums},
    volume = {10},
    issn = {1432-0444},
    doi = {10.1007/BF02573970},
    language = {en},
    number = {2},
    urldate = {2026-03-16},
    journal = {Discrete \& Computational Geometry},
    author = {Barvinok, Alexander I.},
    month = jul,
    year = {1993},
    pages = {123--141},
}

@article{cohenTwoAlgorithmsDetermining1979,
    title = {Two {Algorithms} for {Determining} {Volumes} of {Convex} {Polyhedra}},
    volume = {26},
    issn = {0004-5411},
    doi = {10.1145/322139.322141},
    number = {3},
    urldate = {2026-03-16},
    journal = {J. ACM},
    author = {Cohen, Jacques and Hickey, Timothy},
    month = jul,
    year = {1979},
    pages = {401--414},
}

@article{banach1922operations,
    title = {Sur les opérations dans les ensembles abstraits et leur application aux équations intégrales},
    volume = {3},
    doi = {10.4064/fm-3-1-133-181},
    number = {1},
    journal = {Fundamenta mathematicae},
    publisher = {Polska Akademia Nauk. Instytut Matematyczny PAN},
    author = {Banach, Stefan},
    year = {1922},
    pages = {133--181},
}

@article{brouwerUeberAbbildungMannigfaltigkeiten1911,
    title = {Über {Abbildung} von {Mannigfaltigkeiten}},
    volume = {71},
    issn = {1432-1807},
    doi = {10.1007/BF01456931},
    language = {de},
    number = {1},
    urldate = {2025-01-20},
    journal = {Mathematische Annalen},
    author = {Brouwer, L. E. J.},
    month = mar,
    year = {1911},
    pages = {97--115},
}

@article{condonComplexityStochasticGames1992,
    title = {The {Complexity} of {Stochastic} {Games}},
    volume = {96},
    issn = {0890-5401},
    doi = {10.1016/0890-5401(92)90048-K},
    number = {2},
    urldate = {2023-11-24},
    journal = {Information and Computation},
    author = {Condon, Anne},
    month = feb,
    year = {1992},
    pages = {203--224},
}

@article{shapley1953stochastic,
    title = {Stochastic games},
    volume = {39},
    number = {10},
    journal = {Proceedings of the national academy of sciences},
    publisher = {National Acad Sciences},
    author = {Shapley, L. S.},
    year = {1953},
    pages = {1095--1100},
}

@article{fearnleyUniqueEndPotential2020,
    title = {Unique {End} of {Potential} {Line}},
    volume = {114},
    issn = {00220000},
    doi = {10.1016/j.jcss.2020.05.007},
    language = {en},
    urldate = {2023-09-13},
    journal = {Journal of Computer and System Sciences},
    author = {Fearnley, John and Gordon, Spencer and Mehta, Ruta and Savani, Rahul},
    month = dec,
    year = {2020},
    pages = {1--35},
}

@article{ludwigSubexponentialRandomizedAlgorithm1995,
    title = {A {Subexponential} {Randomized} {Algorithm} for the {Simple} {Stochastic} {Game} {Problem}},
    volume = {117},
    issn = {08905401},
    doi = {10.1006/inco.1995.1035},
    language = {en},
    number = {1},
    urldate = {2023-08-10},
    journal = {Information and Computation},
    author = {Ludwig, W.},
    month = feb,
    year = {1995},
    pages = {151--155},
}

@techreport{bjorklundRandomizedSubexponentialAlgorithms2004,
    title = {Randomized subexponential algorithms for infinite games},
    number = {2004-011},
    institution = {Department of Information Technology, Uppsala University},
    author = {Björklund, Henrik and Sandberg, Sven and Vorobyov, Sergei},
    month = apr,
    year = {2004},
}

@article{halmanSimpleStochasticGames2007,
    title = {Simple {Stochastic} {Games}, {Parity} {Games}, {Mean} {Payoff} {Games} and {Discounted} {Payoff} {Games} {Are} {All} {LP}-{Type} {Problems}},
    volume = {49},
    issn = {1432-0541},
    doi = {10.1007/s00453-007-0175-3},
    language = {en},
    number = {1},
    urldate = {2024-02-17},
    journal = {Algorithmica},
    author = {Halman, Nir},
    month = sep,
    year = {2007},
    pages = {37--50},
}

@article{fearnleyFasterAlgorithmFinding2022,
    title = {A {Faster} {Algorithm} for {Finding} {Tarski} {Fixed} {Points}},
    volume = {18},
    issn = {1549-6325, 1549-6333},
    doi = {10.1145/3524044},
    language = {en},
    number = {3},
    urldate = {2023-08-10},
    journal = {ACM Transactions on Algorithms},
    author = {Fearnley, John and Pálvölgyi, Dömötör and Savani, Rahul},
    month = jul,
    year = {2022},
    pages = {1--23},
}

@incollection{daskalakisContinuousLocalSearch2011,
    series = {Proceedings},
    title = {Continuous {Local} {Search}},
    isbn = {978-0-89871-993-2},
    doi = {10.1137/1.9781611973082.62},
    urldate = {2024-06-12},
    booktitle = {Proceedings of the 2011 {Annual} {ACM}-{SIAM} {Symposium} on {Discrete} {Algorithms} ({SODA})},
    publisher = {Society for Industrial and Applied Mathematics},
    author = {Daskalakis, Constantinos and Papadimitriou, Christos},
    month = jan,
    year = {2011},
    pages = {790--804},
}

@inproceedings{daskalakisConverseBanachFixed2018,
    address = {New York, NY, USA},
    series = {{STOC} 2018},
    title = {A converse to {Banach}'s fixed point theorem and its {CLS}-completeness},
    isbn = {978-1-4503-5559-9},
    doi = {10.1145/3188745.3188968},
    urldate = {2024-06-20},
    booktitle = {Proceedings of the 50th {Annual} {ACM} {SIGACT} {Symposium} on {Theory} of {Computing}},
    publisher = {Association for Computing Machinery},
    author = {Daskalakis, Constantinos and Tzamos, Christos and Zampetakis, Manolis},
    month = jun,
    year = {2018},
    pages = {44--50},
}

@article{sikorskiComputationalComplexityFixed2009,
    title = {Computational complexity of fixed points},
    volume = {6},
    issn = {1661-7746},
    doi = {10.1007/s11784-009-0128-3},
    language = {en},
    number = {2},
    urldate = {2025-01-14},
    journal = {Journal of Fixed Point Theory and Applications},
    author = {Sikorski, Krzysztof},
    month = dec,
    year = {2009},
    pages = {249--283},
}

@article{huangApproximatingFixedPoints1999,
    title = {Approximating {Fixed} {Points} of {Weakly} {Contracting} {Mappings}},
    volume = {15},
    issn = {0885-064X},
    doi = {10.1006/jcom.1999.0504},
    number = {2},
    urldate = {2024-06-17},
    journal = {Journal of Complexity},
    author = {Huang, Z. and Khachiyan, L. and Sikorski, K.},
    month = jun,
    year = {1999},
    pages = {200--213},
}

@article{sikorskiEllipsoidAlgorithmComputation1993,
    title = {An {Ellipsoid} {Algorithm} for the {Computation} of {Fixed} {Points}},
    volume = {9},
    issn = {0885-064X},
    doi = {10.1006/jcom.1993.1013},
    number = {1},
    urldate = {2024-06-18},
    journal = {Journal of Complexity},
    author = {Sikorski, K. and Tsay, C.W. and Woźniakowski, H.},
    month = mar,
    year = {1993},
    pages = {181--200},
}

@article{shellmanTwoDimensionalBisectionEnvelope2002,
    title = {A {Two}-{Dimensional} {Bisection} {Envelope} {Algorithm} for {Fixed} {Points}},
    volume = {18},
    issn = {0885-064X},
    doi = {10.1006/jcom.2001.0625},
    number = {2},
    urldate = {2024-06-18},
    journal = {Journal of Complexity},
    author = {Shellman, Spencer and Sikorski, K.},
    month = jun,
    year = {2002},
    pages = {641--659},
}

@article{shellmanRecursiveAlgorithmInfinitynorm2003,
    title = {A recursive algorithm for the infinity-norm fixed point problem},
    volume = {19},
    issn = {0885-064X},
    doi = {10.1016/j.jco.2003.06.001},
    number = {6},
    urldate = {2025-01-14},
    journal = {Journal of Complexity},
    author = {Shellman, Spencer and Sikorski, K.},
    month = dec,
    year = {2003},
    pages = {799--834},
}

@inproceedings{haslebacherQueryEfficientFixpointsLpContractions2025,
    title = {Query-{Efficient} {Fixpoints} of $\ell_p$-{Contractions}},
    issn = {2575-8454},
    doi = {10.1109/FOCS63196.2025.00109},
    urldate = {2026-03-18},
    booktitle = {2025 {IEEE} 66th {Annual} {Symposium} on {Foundations} of {Computer} {Science} ({FOCS})},
    author = {Haslebacher, Sebastian and Lill, Jonas and Schnider, Patrick and Weber, Simon},
    month = dec,
    year = {2025},
    note = {ISSN: 2575-8454},
    pages = {2037--2053},
}

@book{edelsbrunnerAlgorithmsCombinatorialGeometry1987,
    address = {Berlin, Heidelberg},
    title = {Algorithms in {Combinatorial} {Geometry}},
    copyright = {http://www.springer.com/tdm},
    isbn = {978-3-642-64873-1 978-3-642-61568-9},
    doi = {10.1007/978-3-642-61568-9},
    urldate = {2026-03-28},
    publisher = {Springer},
    author = {Edelsbrunner, Herbert},
    year = {1987},
}

@incollection{edelsbrunnerConstructingArrangements1987,
    address = {Berlin, Heidelberg},
    title = {Constructing {Arrangements}},
    isbn = {978-3-642-61568-9},
    doi = {10.1007/978-3-642-61568-9_7},
    language = {en},
    urldate = {2026-03-28},
    booktitle = {Algorithms in {Combinatorial} {Geometry}},
    publisher = {Springer},
    author = {Edelsbrunner, Herbert},
    editor = {Edelsbrunner, Herbert},
    year = {1987},
    pages = {123--137},
}

@inproceedings{batziouMonotoneContractions2025,
    address = {New York, NY, USA},
    series = {{STOC} '25},
    title = {Monotone {Contractions}},
    isbn = {979-8-4007-1510-5},
    doi = {10.1145/3717823.3718175},
    urldate = {2026-03-31},
    booktitle = {Proceedings of the 57th {Annual} {ACM} {Symposium} on {Theory} of {Computing}},
    publisher = {Association for Computing Machinery},
    author = {Batziou, Eleni and Fearnley, John and Gordon, Spencer and Mehta, Ruta and Savani, Rahul},
    month = jun,
    year = {2025},
    pages = {507--517},
}

\appendix

\section{Reproduced Proof of \Cref{thm:centerpoint_theorem}}
\label{appendix:reproduced_centerpoint_proof}

The proof of \Cref{thm:centerpoint_theorem} in~\cite{chenComputingFixedPoint2025} is formulated for specific types of pointsets, but it also works for our notion of centerpoints defined via volume (see Lemma~6 and Lemma~7 in~\cite{chenComputingFixedPoint2025}). For convenience, we decided to sketch the proof on a high level here (in our setting).

\thmcenterpoint*

\begin{proof}
    Let $X \subseteq [0, 1]^d$ be arbitrary with $\vol(X) \geq 0$. Recall the notion of pyramids from \Cref{ssec:proof_thm_1}. We will proceed by showing that there exists a so-called balanced point $c \in [0, 1]^d$, i.e.\ a point that satisfies $\vol(P^+_i(c) \cap X) = \vol(P^-_i(c) \cap X)$ for all $i \in [d]$. By the pyramid decomposition of $\ell_\infty$-halfspaces (\Cref{obs:pyramid_decomposition}), it then follows that $c$ must be a $\frac{1}{2}$-centerpoint, as desired.

    To prove existence of a balanced point, we use Brouwer's fixed point theorem, which says that a continuous function $f : K \rightarrow K$ from a non-empty convex compact subset $K \subseteq \R^d$ to itself has a fixed point $z = f(z)$. We construct $f$ as follows: Given a point $z \in \R^d$, we start by defining 
    \[
        f(z)_i = z_i + \delta \left( \vol(P^+_i(z) \cap X) - \vol(P^-_i(z) \cap X) \right)
    \]
    for all $i \in [d]$ and some fixed parameter $\delta > 0$. Clearly, $f$ is continuous. 
    
    It remains to restrict $f$ to the cube $K = [0, 1]^d$ while making sure that it cannot escape, which can be achieved by choosing the parameter $\delta > 0$ small enough such that $f([0, 1]^d) \subseteq [0, 1]^d$ is guaranteed. Indeed, such a suitably small $\delta$ exists and can e.g.\ be computed by using our \Cref{lemma:precision} (precision lemma).

    With this suitable choice of $\delta$, we thus have a continuous function $f : [0, 1]^d \rightarrow [0, 1]^d$, which by Brouwer's fixed point theorem must have a fixed point $c \in [0, 1]^d$. By definition of $f$, $c$ must be balanced and therefore a $\frac{1}{2}$-centerpoint of $X$.
\end{proof}

\section{Reproduced Proof of \Cref{lemma:ball_contained}}
\label{app:reproving}

For completeness, we provide a full proof of \Cref{lemma:ball_contained} here. Note that this argument is adapted from Lemma IV.1 in~\cite{haslebacherQueryEfficientFixpointsLpContractions2025}.

\lemmaballcontained*

\begin{proof}
    Assume for a contradiction that there exists $x \in B^\infty_r(x^\star) \cap H_{c - f(c)}(c)$. By $x \in B^\infty_r(x^\star)$ we know $\|x - x^\star\| \leq r$, and by $x \in H_{c - f(c)}(c)$ we know $\|x - f(c)\| \geq \|x - c\|$. Thus, we can calculate
    \[
        \|c - x^\star\| \leq \|c - x \| + r \leq \|f(c) - x\| + r \leq \|f(c) - x^\star\| + 2r \leq \lambda \|c - x^\star\| + 2r 
    \]
    by repeatedly using the triangle inequality, and in the final step the contraction property. This implies $\|c - x^\star\| \leq \frac{2r}{1 - \lambda}$ and thus also $\|f(c) - x^\star\| \leq \lambda \frac{2r}{1 - \lambda}$. Using the triangle inequality again, we therefore obtain $\|c - f(c)\| \leq (1 + \lambda) \frac{2r}{1 - \lambda} = \epsilon$, a contradiction.
\end{proof}

\end{document}